\documentclass[12pt, draftclsnofoot, onecolumn]{IEEEtran}

\usepackage[T1]{fontenc}

\usepackage[british]{babel}
\usepackage{hyphenat}
\usepackage{booktabs} 
\usepackage{graphicx,multicol,latexsym,amsmath,amssymb}
\usepackage{color}
\usepackage[dvipsnames]{xcolor}
\usepackage{amsthm}
\usepackage{subcaption}
\newcommand{\beq}{\begin{equation}}
\newcommand{\eeq}{\end{equation}}

\newtheorem{theorem}{Theorem}[section]

\newtheorem{lemma}{Lemma}

\newtheorem{remark}{Remark}

\newif\ifnotesw
\noteswtrue

\newif\ifnotesw
\noteswtrue

\usepackage{amsmath}
\usepackage[cmintegrals]{newtxmath}
\begin{document}

\title{Connecting flying backhauls of UAVs to enhance vehicular networks with fixed 5G NR infrastructure~\footnote{Part of this work was presented at WiSARN'2020, IEEE Conference on Computer Communications Workshops (INFOCOM WKSHPS), International Workshop on Wireless Sensor, Robot and UAV Networks, July 6th, 2020, (Online). Connecting flying backhauls of drones to enhance vehicular networks with fixed 5G NR infrastructure, P. Jacquet D. Popescu and B. Mans, IEEE Press, pp. 472-477, doi: 10.1109/INFOCOMWKSHPS50562.2020.9162670.}}
%
%
%

\author{Dalia~Popescu, 
        Philippe~Jacquet,~\IEEEmembership{Fellow,~IEEE,}
        Bernard~Mans
        
\thanks{B.  Mans  was  supported  in  part  by  the  Australian  Research  Council  under Grant DP170102794.

Part of the work has been done at Lincs, Paris, France.
Dalia Popescu was with Nokia Bell Labs, France. Philippe Jacquet is with INRIA, France. 
Bernard Mans is with Macquarie University, Sydney, Australia (e-mail:  bernard.mans@mq.edu.au).}
}

%
%


\maketitle

\begin{abstract}
This paper investigates moving networks of Unmanned Aerial Vehicles (UAVs), such as drones,  as one of the innovative opportunities brought by the 5G. With a main purpose to extend connectivity and guarantee data rates, the drones require an intelligent choice of hovering locations due to their specific limitations such as flight time and coverage surface. 
To this end, we provide analytic bounds on the requirements in terms of connectivity extension for vehicular networks served by fixed Enhanced Mobile BroadBand (eMBB) infrastructure, where both vehicular networks and infrastructures are modeled using stochastic and fractal geometry as a macro model for urban environment, providing a unique perspective into the smart city. 

Namely, we prove that assuming $n$ mobile nodes (distributed according to a hyperfractal distribution of dimension $d_F$) and an average of $\rho$ Next Generation NodeB (gNBs), distributed like an hyperfractal of dimension $d_r$ if $\rho=n^\theta$ with $\theta>d_r/4$ and letting $n$ tending to infinity (to reflect megalopolis cities), then the average fraction of mobile nodes not covered by a gNB tends to zero like $O\left(n^{-\frac{(d_F-2)}{d_r}(2\theta-\frac{d_r}{2})}\right)$. Interestingly, we then prove that the average number of  drones,  needed  to connect each  mobile  node not  covered  by gNBs  is comparable to the number of isolated mobile nodes. We complete the characterisation by proving that when $\theta<d_r/4$ the proportion of covered mobile nodes tends to zero.


Furthermore, we provide insights on the intelligent placement of the ``garage of drones'',  the  home  location of  these  nomadic  infrastructure  nodes,  such as to minimize what we call the ``flight-to-coverage time''. We  provide a fast procedure to select the relays that will be garages (and store drones) in order to minimize the number of garages and minimize the delay. Finally we confirm our analytical results using simulations carried out in Matlab.
\end{abstract}
\maketitle

%

\IEEEpeerreviewmaketitle

\section{Introduction}

\subsection{Mobile networks}
{\em 5G New Radio (5G NR)} is envisioned to offer a diverse palette of services to mobile users and platforms, with often incompatible types of requirements. 5G will support: (i) the {\em enhanced Mobile BroadBand (eMBB)} applications with throughput of order of 100 Mbit/s per user, (ii) the {\em Ultra Reliable Low Latency Communications (URLLC)} for industrial and vehicular environments with the hard constraint of 1 ms latency and 99.99 $\%$ reliability, and (iii) the {\em massive Machine Type Communications (mMTC)} with a colossal density in order of 100 per square km. 5G will achieve this through a new air interface and novel network architectures that will either be evolved from the current 4G systems, or be completely drawn from scratch. 

However, it is clear that by the time the 5G standards start to be deployed, many of these important challenges will still remain opened and a truly disruptive transition from current 4G systems will only occur over time. One of these challenges is the ability of the network to adapt efficiently to traffic demand evolution in space and time, in particular when new frequency bands in a range much higher than ever dared before are to be exploited, e.g. Frequency Range (FR) 4 at over 52.6 GHz.
\begin{figure}\centering
\vspace{0.2cm}\centering
\includegraphics[scale=0.42, trim=0cm 2cm 2cm 0cm]{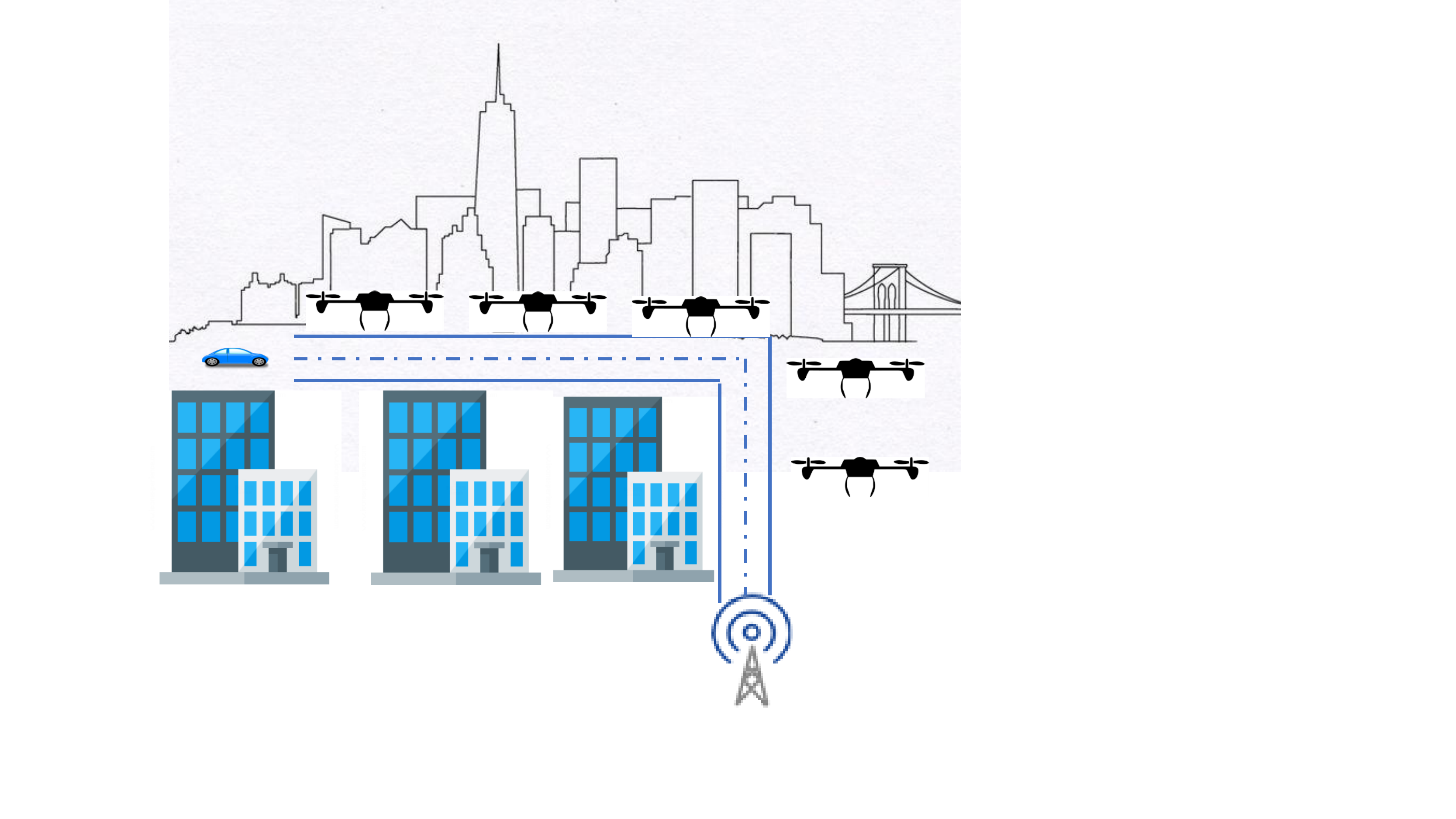}
\caption{Extension of connectivity by drones}
\label{fig:drones_connectivity}
\end{figure}
A key difference from today's 4G systems is that 5G networks will be characterized by a massive density of nodes, both human-held and machine-type: according to the Ericsson mobility report (June 2020), 5G subscriptions are forecast to reach 2.8 billion globally by the end of 2025. A larger density of wireless nodes implies a larger standard deviation in the traffic generation process. Network planning done based on average (or peak) traffic predictions as done for previous systems, will only be able to bring sub-optimal results. Moreover, the ``{\it verticals}'' of major interest identified for 5G networks (automotive, health, fabric-of-the-future, media, energy), come with rather different requirements and use cases. Future networks will handle extremely heterogeneous traffic scenarios.

To tackle these problems, 5G NR networks are to exhibit a network flexibility that is much higher than in the past: infrastructure nodes must be adaptable enough in order to be able to smoothly and autonomously react to the fast temporal and spatial variations of traffic demand. 
The level of flexibility that can be achieved through such advances still faces a fundamental limit: hardware location is static and the offered network capacity on a local scale is fundamentally limited by the density of the infrastructure equipment (radio transceivers) in the area of interest \cite{gall2019relayassisted}.

This opens up the possibility of {\em Moving Networks}, informally defined as moving nodes, with advanced network capabilities, gathered together to form a movable network that can communicate with its environment. Moving Networks will help 5G systems to become demand-attentive, with a level of network cooperation that will facilitate the provisioning of services to users characterized by high mobility and throughput requirements, or in situations where a fixed network cannot satisfy traffic demand. 
A key player in a moving network is the communication entity with the highest number of degrees of freedom of movement: the drone.

\subsection{Contributions}
In this work we initiate the study and design of moving networks by a first analysis of the provisioning and dimensioning of a network where drones act as flying backhaul. The setup we make use is that of a smart city in which we design one of the complex 5G NR urban scenarios: drones coexist with vehicular networks and fixed telecommunication infrastructures. 

Making use of the innovative model called Hyperfractal \cite{TON_HF} for representing the traffic density of vehicular devices in an urban scenario and the distribution of static telecommunication infrastructure, we derive the requirements in terms of resources of unmanned aerial vehicles (UAVs) for enhancing the coverage required by the users. 
Informally, we compute the expected percentage of the users in poor coverage conditions and derive the average number of drones necessary for ensuring the coverage with high reliability of the vehicular nodes. Furthermore, we discuss the notion of "garage of drones", the home location of these nomadic infrastructure nodes, the drones.

More specifically, our contributions are:
\begin{itemize}
    \item An innovative macro model for urban environment, providing a new perspective on smart cities,
    where both vehicular networks and fixed eMBB infrastructure networks are modeled using stochastic and fractal geometry (Section~\ref{model}). The model allows to compute precise bounds on the requirements in term of connectivity, or lack of thereof.
    \item A proof that the average fraction of mobile nodes not covered by the fixed infrastructure (gNBs) either tends to zero, or tends to the actual number of nodes in the network with the characterisation of the exact threshold  (Theorem~\ref{without_drones} and Theorem~\ref{theo:without2} in Section~\ref{results}).
    \item A proof that  the number  of drones  to  connect  the  isolated  mobile  nodes  is  asymptotically equivalent  to  the  number  of  isolated  mobile  nodes (Theorem~\ref{theodrone} in Section~\ref{results}). Thus, proving that, typically, one drone is sufficient to connect an isolated mobile node to the network.
    \item An introduction and analysis of the ``flight-to-coverage time'' for the deployed drones, helping in their optimised placement and recharge in the network (Section~\ref{garage}). We  provide a fast procedure to select the relays that will be garages (and store drones) in order to minimize the number of garages and minimize the delay.
    \item Simulations results in Matlab that
    confirm our stochastic results (Section~\ref{simulations}).
\end{itemize}

\subsection{Related works}
While their introduction to commercial use has been delayed and restricted to specific use-cases due to the numerous challenges (e.g.~\cite{survey1}), the flexibility the drones bring to the network planning by their increased degree of movement, has motivated the industry to push towards their introduction in wide scenarios of 5G. 

We refer the reader to recent relevant surveys. For  instance, a thorough Tutorial on UAV Communications for 5G and Beyond is presented in~\cite{QingqingWu_uavTuto}, where both UAV-assisted wireless communications and cellular-connected UAVs are discussed (with UAVs integrated into the network as new aerial communication platforms and users).

Importantly, many new opportunities have been often highlighted (e.g.~in \cite{survey2}), including (not exhaustively):
\begin{itemize}
    \item Coverage and capacity enhancement of Beyond 5G wireless cellular networks:
    \item UAVs as flying base stations for public safety scenarios;
    \item UAV-assisted terrestrial networks for information dissemination;
    \item Cache-Enabled UAVs;
    \item UAVs as flying backhaul for terrestrial networks.
\end{itemize}

Considerations for a multi-UAV enabled wireless communication system, where multiple UAV-mounted aerial base stations are employed to serve a group of users on the ground have been presented in~\cite{QingqingWu_uavs}.

An overview of UAV-aided wireless communications, introducing the basic networking architecture and main channel characteristics, and highlighting the key design considerations as well as the new opportunities to be exploited is presented in~\cite{YongZeng_survey}.

Grasping the interest for this new concept of communication, the research community has started analyzing the details of the new communication paradigms introduced with the drones. 
In the beginning, drones have been considered for delivering the capacity required for sporadic peaks of demand, such as in entertainment events, to reach areas where there is no infrastructure or the infrastructure is down due to a natural catastrophe. 
Yet in the new 5G NR scenarios of communication, drones are not only used in isolated cases but are considered as active components in the planning of moving networks in order to provide the elasticity and flexibility required by these. In this sense, using drones for the 5G Integrated Access and Backhaul (IAB) (\cite{IAB, IAB_3GPP}) is one of the most relevant use-case scenarios, ensuring the long-desired flexibility of coverage aimed with an adaptable cost of infrastructure.


The research community has been looking at specific problems in wireless networks employing drones.
In \cite{drones1}, the authors prove the feasibility of multi-tier  networks of drones giving some first insight on  the dimensioning of such a network. 
The work done in \cite{c7} uses two approaches, a network-centric approach and a user-centric approach to find the optimal backhaul placement for the drones.  
On the other hand, in \cite{drones_3D}, the authors propose a heuristic algorithm for the placement of drones. %
The authors of \cite{drones2} analyze the scenario of a bidirectional highway and propose an algorithm for determining the number of drones for satisfying coverage and delay constraints. A delay analysis is performed in~\cite{drones3} on the same type of scenario. Simplified assumptions are used in \cite{drones4} to showcase the considerable improvements brought by UAVs as assistance to cellular networks. The movement of network of drones, so-called swarm is analyzed in \cite{drones6} and efficient routing techniques are proposed.
Drones are so much envisioned for the networks of the future that the authors of \cite{drones_as} provide solutions for a platform of drones-as-a-service that would answer the future operators demands.

Following these observations from the state of the art, we aim to initiate a complex analysis of the use of the drones in a smart city for different purposes. To our knowledge, this is the first analysis which looks at the entire macro urban model of the city, incorporating both devices and the fixed telecommunications infrastructure. In our analysis, the drones are used as a flying backhaul for extending the coverage to users in poor conditions.

\section{System model and scenario}\label{model}
The communication scenario in our work has the aim to provide a flexible network architecture that allows serving the vehicular devices with tight delay constraints. The scenario comprises three types of communication entities: the vehicles which we denote as the user equipment (UEs), the fixed telecommunication infrastructure called Next Generation NodeB (gNBs) and the moving network nodes which are the drones also called UAVs. 

The scenario we tackle in this work is as follows. 
An urban network of vehicles is served by a fixed telecommunication infrastructure. Due to the limited coverage capability using millimeter wave (mmwave) in urban environments, the costs of installing fixed telecommunication infrastructure throughout the entire city and the mobility of the vehicles in the urban area, some users will be outside of the coverage areas. This is no longer acceptable in 5G as International Mobile Telecommunications Standards (IMT 2020 \cite{IMT}) require for most of the users and in particular for URLLC users, full coverage and harsh constraints for delay. UAVs (such as drones) will therefore be dynamically deployed for reaching the so-called "isolated" users, forming the flying backhaul of the network. 

\subsection{Communication Model}

For the sake of simplicity, we consider the communication to be done in the FR2 (frequency range 2) 28GHz frequency band using beamforming, mmwave technology in half-duplex mode. This follows the current specifications, yet we foresee the modelling to be easily extended for FR3 (frequency range 3) and FR4 (frequency range 4). This leads to highly directive beams with a narrow aperture, sensitive to blockage and interference. To this end, in our modeling we integrate the {\it canyon} propagation model which implies that the signal emitted by a mobile node propagates only on the street where it stands on.

We consider that drones, gNBs and UEs use the same frequency band and transmission power. The transmission range is $R_n=\frac{1}{\sqrt n}$, where $n$ is the number of UEs in the city. 
The reasoning behind this choice of transmission range is as follows. 
The population of a city (in most of the cases) is proportional to the area of the city~\cite{oecd} and the population of cars is proportional to the population of the city (in fact the local variations of car densities counter the local variation of population densities~\cite{emta}), therefore the population of cars is proportional to the area of the city, Area$=A \cdot n$ where $A$ is a constant. A natural assumption is that the absolute radio range is constant. But since we assume in our model that the city map is always a unit square, the relative radio range in the unit square must be $R_n=\frac{R}{\sqrt{An}}$ which we simplify in $R_n=\frac{1}{\sqrt{n}}$. This leads to a limited radio range, that together with the canyon effect supports the features of the mmwave communications.

The mobile nodes that we treat in this modeling are vehicular devices or UEs located on  streets.  
The UAV will communicate with the mobile users on the RAN interface, similar to a gNB to UE exchange. For drone-to-drone communication and drone-to-gNB communication, it is the Integrated Access and Backhaul (IAB) interface that is used, in the same frequency band.

We consider that, for a UE to benefit from the dynamical coverage assistance provided by UAVs, it has to be already registered to the network through a previously performed random access procedure (RACH) and we only treat users in connected mode, therefore the network is aware about the existence of the UE, its context and service requirements and, when the situation arrives, that it is experiencing a poor coverage. 

The communication scenario exploits the fact that the position of a registered UE in connected mode is known to the network and the evolution can be tracked (speed, direction of movement). The position of the UE can be either UE based positioning or UE assisted positioning, both allowing enough degree of accuracy for a proactive preparation of resources in the identified target cell. These are assumptions perfectly inline with the objectives of Release 17 of 5G NR. 

As a consequence of this knowledge in the network, the gNB that is currently serving the vehicle can inform the target gNB about the upcoming arrival of the vehicle such that the target gNB can prepare/send a drone (if necessary) for ensuring the service continuity to the UE. 

The drones will be dynamically deployed in order to extend the coverage of the gNBs towards the UEs, forming a Flying Backhaul. 

\subsection{Hyperfractal Modeling for an urban scenario}
As 5G is not meant to be designed for a general setup but rather have specific solutions for all the variety of communication scenarios (in particular for {\it verticals}), the idea of a smart city
in which one is capable of translating different parameters and features in order to better calibrate the private network being deployed arises as a possible answer to these stringent requirements. 

An hyperfractal representation of urban settings~\cite{TON_HF}, in particular for mobile vehicles and fixed infrastructures (such as red lights), is an innovative representation that we chose for modelling the UEs and eMBB infrastructures here. With this model, one is capable of taking measurements of vehicular traffic flows, urban characteristics (streets lengths, crossings, etc), fit the data to a hyperfractal with computed parameters and compute metrics of interest. Independently, this is an interesting step towards achieving modeling of smart urban cities that can be exploited in other scenarios. 

\subsubsection{Mobile users}

The positions of the mobile users and their flows in the urban environment are modeled with the hyperfractal model described in \cite{TON_HF, lcn_us, mswim_us}. In the following, we only provide the necessary and self-sufficient introduction to the model but a complete and extended description can be found in \cite{TON_HF}. 

The map model lays in the unit square embedded in the 2-dimensional Euclidean space.
The support of the population is a grid of streets. Let us denote this structure by $\mathcal{X}=\bigcup_{l=0}^\infty \mathcal{X}_l$.  A street of level $H$ consists of the union of consecutive segments of level $H$ in the same line. The length of a street is the length of the side of the map. 

The mobile users are modeled by means of a 
Poisson point process $\Phi$ on $\mathcal{X}$
with total intensity (mean number of points) $a$ ($0<a<\infty$)
having 1-dimensional intensity 
\beq
\lambda_l=a(p/2)(q/2)^l
\label{eq:dens_mobiles}
\eeq
on $\mathcal{X}_l$, $l=0,\ldots,\infty$,
with $q=1-p$ for some parameter $p$ ($0\le p\le 1$).
Note that $\Phi$ 
can be constructed in the following way: one samples the total  number of mobiles users $\Phi(\mathcal{X})=n$ from Poisson$(a)$ distribution; 
each mobile is placed independently 
with probability $p$ on $\mathcal{X}_0$ according to the uniform distribution
and with probability $q/4$ it is recursively located in the similar way in one the four quadrants of $\bigcup_{l=1}^\infty \mathcal{X}_l$.

 

The intensity  measure of  $\Phi$ on  $\mathcal{X}$ is hypothetically reproduced in each of the four quadrants of $\bigcup_{l=1}^\infty \mathcal{X}_l$ with the scaling of its support by the factor 1/2 and of its value by $q/4$.

The fractal dimension is a scalar parameter characterizing a geometric object with repetitive patterns. It indicates how the volume of the object decreases when submitted to an homothetic scaling. When the object is a convex subset of an euclidian space of finite dimension, the fractal dimension is equal to this dimension. When the object is a fractal subset of this euclidian space as defined in \cite{mandelbrot}, it is a possibly non integer but positive  scalar strictly smaller than the euclidian dimension. When the object is a measure defined in the euclidian space, as it is the case in this paper, then the fractal dimension can be strictly larger than the euclidian dimension. In this case we say that the measure is {\it hyper-fractal}. 

 \begin{remark}
 The fractal dimension $ d_F$ of the intensity measure of $\Phi$ satisfies
\begin{equation*} \label{eq:d_F}
\left(\frac{1}{2}\right)^{d_F}=\frac{q}{4} \qquad\text{thus}\qquad d_F=\frac{\log(\frac{4}{q})}{\log 2}\ge 2.
\end{equation*}
\end{remark}
The fractal dimension $d_F$ is greater than~2, the Euclidean dimension of the square in which it is embedded, thus 
the model was called in~\cite{spaswin} {\em hyperfractal}.
Figure \ref{fig:map_support} shows an example of support iteratively built up to level $H=3$ while Figure \ref{fig:hyper} shows the nodes obtained as a Poisson shot on a support of a higher depth, $H=5$.
\hspace*{2cm}
\begin{figure}  [ht!]
\hspace*{2cm}
\begin{subfigure}[t]{0.225\textwidth}
\includegraphics[scale=0.27]{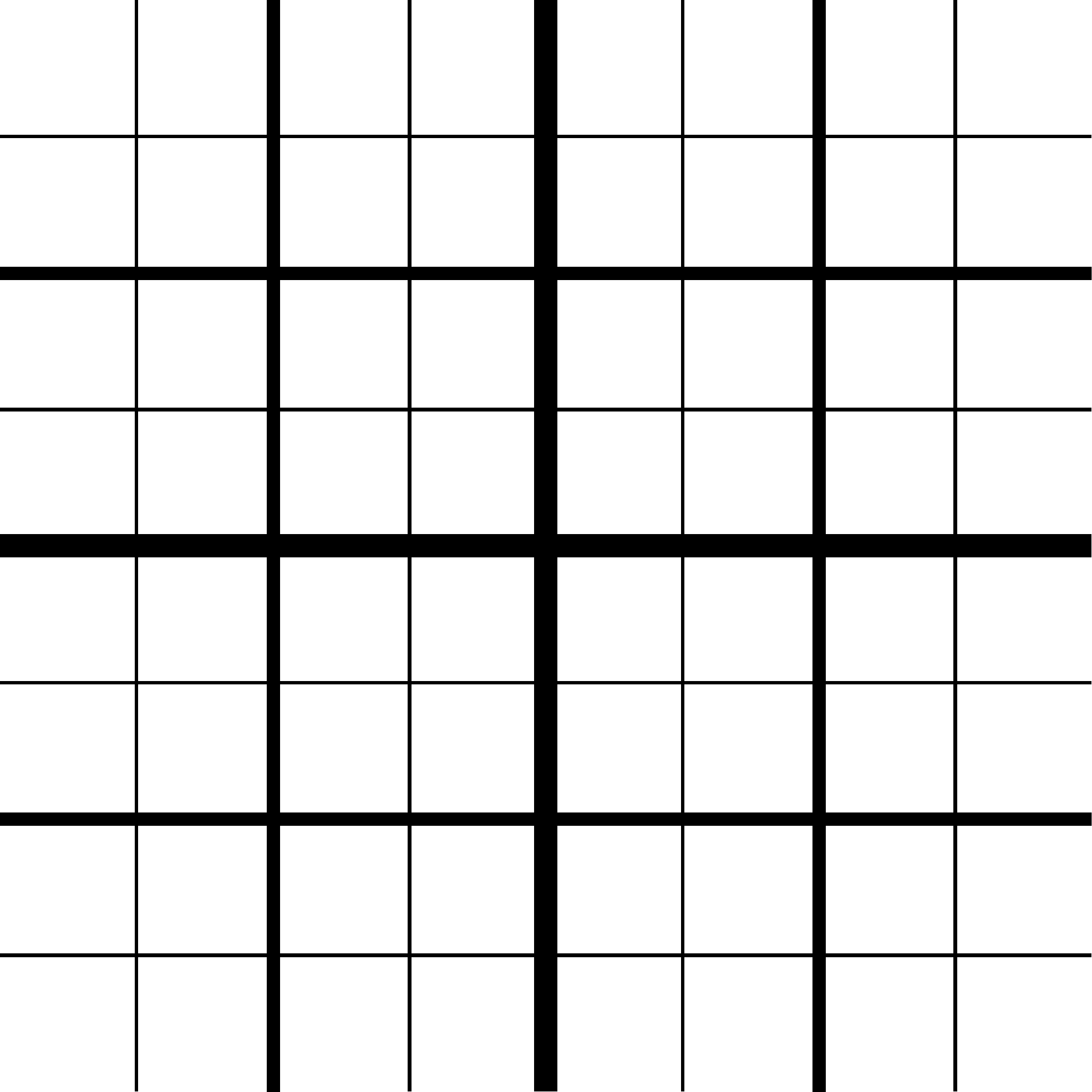}
\caption{}
\label{fig:map_support}
\end{subfigure}
\hspace*{2cm}
\begin{subfigure}[t]{0.225\textwidth}
\includegraphics[scale=0.53,  trim={3cm 9.5cm 2cm 18cm}]{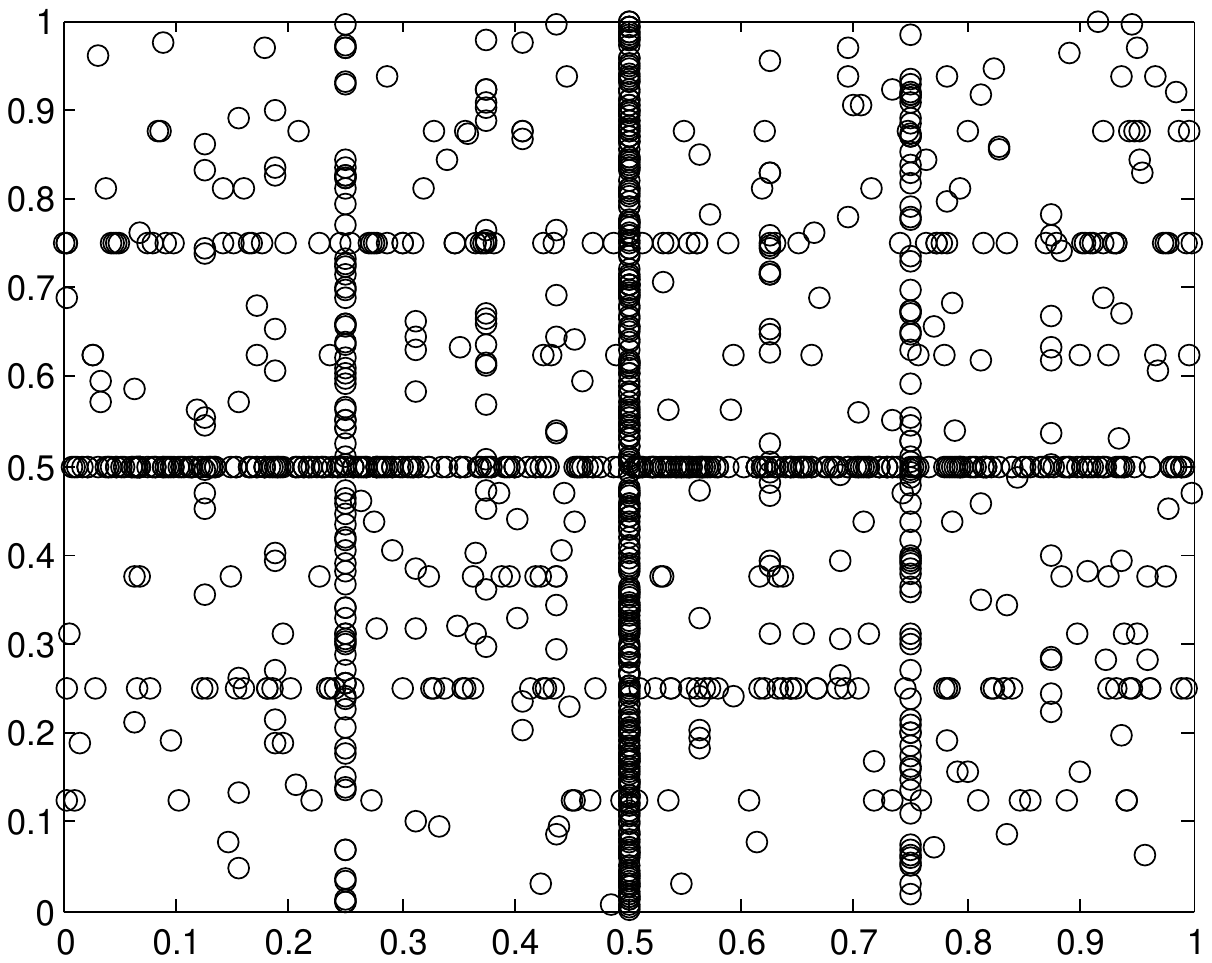}
\caption{}
\label{fig:hyper}
\end{subfigure}
\caption{(a) Hyperfractal map support; (b) Hyperfractal, $d_F=3$, $n=1,200$ nodes;}
\end{figure}
\subsubsection{Fixed telecommunication infrastructure, gNBs}
We denote the process for gNBs by $\Xi$. To define $\Xi$ it is convenient to consider an auxiliary  Poisson process $\Phi_r$ with both processes supported by a 1-dimensional subset of $\mathcal{X}$ namely, the set of intersections of segments constituting $\mathcal{X}$. We assume that $\Phi_r$ has discrete intensity   
\begin{equation}\label{eq:intersection}
p(h,v)=\rho(p')^2\left(\frac{1-(p')}{2}\right)^{h+v}
\end{equation}
on all intersections $\mathcal{X}_h\cap\mathcal{X}_v$
for $h,v=0,\ldots,\infty$ for some parameter $p'$, $0\le p'\le 1$ and $\rho>0$. That is, 
on any such intersection the mass of $\Phi_r$ is Poisson random variable with parameter $p(h,v)$ and $\rho$
is the total expected number of points of $\Phi_r$ in the model.
The self-similar structure of $\Phi_r$ is well explained by its construction
in which we first sample the total number of points from a Poisson distribution
of intensity $\rho$ and given $\Phi_r(\mathcal{X})=M$,
each point is independently placed with probability $(p')^2$ in the central crossing of  $\mathcal{X}_0$, with  probability $p'\frac{1-p'}{2}$  on some other crossing of one  of the four segments forming  $\mathcal{X}_0$ and, with the remaining probability $\left(\frac{1-(p')}{2}\right)^2$, in a similar way, recursively, on some crossing of one of the four quadrants of $\bigcup_{l=1}^\infty \mathcal{X}_l$.
 

Note that the Poisson process $\Phi_r$ is not simple and 
we define the process $\Xi$ for gNBs as the support measure of $\Phi_r$, i.e., 
only one gNB is installed in every crossing where $\Phi_r$ has at least one point.

\begin{remark}
Note that fixed infrastructure $\Xi$ forms a non-homogeneous binomial point process (i.e. points are placed independently) on the crossings of $\mathcal{X}$ with a given intersection of two segments from $\mathcal{X}_h$ and $\mathcal{X}_v$ occupied by a gNB point with probability $1-\exp(-\rho p(h,v))$.
\end{remark}
Similarly to the process of mobile nodes, we can define the fractal dimension of the gNBs process.
\begin{remark}
The fractal dimension $d_r$ of the probability density of $\Xi$ is equal to the fractal dimension of the intensity measure of the Poisson process $\Phi_r$ and verifies 
\begin{equation*} \label{eq:d_r}
d_r=2\frac{\log(2/(1-p'))}{\log 2}.
\end{equation*}
\end{remark}

\section{Main Results}\label{results}



We recall the strong requirement that 
the vehicles should be covered in a proportion of $\gamma_{FC}>99.99\%$.
In an urban vehicular network supported by eMBB infrastructure, both modeled by hyperfractals, the authors of \cite{gsi}, have proven that, under no constraints on the transmission range, the giant component tends to include all the nodes for $n$ large. 

In this analysis, under the constraint of a transmission range limited to  $R_n=\frac{1}{\sqrt n}$, there will always be disconnected nodes and the number of gNBs to guarantee $\gamma_{FC}$ would not be cost effective. Therefore, we dynamically deploy drones, to extend by hop-by-hop communication, the coverage of the gNBs, as illustrated in Figure~\ref{fig:drones_connectivity}. 

 This should be done by considering the constraints in latency as well, which, in our model, translate to a maximum number of drones on which the packet is allowed to hop before reaching the gNB.


\subsection{Connectivity with fixed telecommunication infrastructure}

We assume that the infrastructure location can be fitted to a hyperfractal distribution with dimension $d_r$ and intensity $\rho$ (to provide a realistic model~\cite{mswim_us}). In a first instance, we omit the existence of drones. Therefore a mobile node is covered only when a gNB exists at distance lower than $R_n=\frac{1}{\sqrt{n}}$. 
In this case, the following theorem gives the average number of nodes that are not in the coverage range of the infrastructure. 

\begin{theorem}[Connectivity without drones]\label{without_drones}
Assuming $n$ mobile nodes distributed according to a hyperfractal distribution of dimension $d_F$ and a gNB distribution of dimension $d_r$ and intensity $\rho$, if $\rho=n^\theta$ with $\theta>d_r/4$ then the average number of mobile nodes not covered by a gNB is $nI_n$ with $I_n=O\left(n^{-\frac{(d_F-2)}{d_r}(2\theta-d_r/2)}\right)$.
\label{theorelay}\end{theorem}

\begin{remark} The quantity $I_n$ is the probability that a mobile node is isolated and thus the theorem proves that the proportion of non covered mobile nodes tends to zero.
\end{remark}
\begin{proof}
Let a mobile node be on a road  of level $H$ at the abscissa $x$. For the mobile node to be covered, it is necessary for a gNB to exist on the same street in the interval $\big[x-\frac{1}{\sqrt n},x+\frac{1}{\sqrt n} \big]$.  

We shall first estimate the probability $I_n(x,H)$ that a mobile node is not covered. Let $N_V(x)$ be the number of intersections of level $V$ ({\it i.e.} with a street of level $V$) which are within distance $R_n$ of abscissa $x$ on the road of level $H$. Since the placement of gNBs follows a Poisson process of mean $\rho p(H,V)$: 
\beq
I_n(x,H)=\exp\left(-\rho\sum_V N_V(x)p(H,V)\right).
\eeq
In order to get an upper bound $I_n(H)$ on $I_n(x,H)$, a lower bound of the quantities $N_V(x)$ is necessary. For any given integer $V\ge 0$, the intersections of level equal or larger than $V$ are regularly spaced at frequency $2^{V+1}$ on the axis. Only half of them are exactly of level $V$, regularly spaced at frequency $2^V$ (as a consequence of the construction process). Therefore a lower bound of $N_V(x)$ is $\lfloor\frac{2^V}{\sqrt{n}}\rfloor$.

Let $V_n$ be the smallest integer which satisfies $2^{V_n}\ge\sqrt{n}$. Therefore the lower bound of $N_H(x)$ is $2^{V-V_n}$ when $V\ge V_n$ and 0 when $V<V_n$. Thus the upper bound $I_n(H)$ of $I_n(x,H)$ is
$$ 
I_n(H)=\exp\left( -\rho\sum_{V\ge V_n}2^{V-V_n}p(H,V)\right)
$$ 
with the expression $p(H,V)=(p')^2(q'/2)^{H+V}$ we get 
\beq
I_n(H)=\exp\left(-\rho p'(q'/2)^{V_n+H}\right).
\eeq
We now have $V_n\ge\log_2 n$ thus $(q'/2)^{V_n}\le n^{-dr/4}$ and with $\rho=n^\theta$ we obtain $I_n(H)\le\exp(-n^{\theta-d_r/4}p'(q'/2)^H)$. From here we can finish the proof since the proportion $I_n$ of isolated mobile nodes is:
\begin{eqnarray} \nonumber
I_n&=&2\sum_H\lambda_H I_n(H)\\ \nonumber
&=&\sum_Hpq^H I_n(H)\\ \label{tight}
&\le& \sum_H pq^H\exp\left(-p'(q'/2)^Hn^{\theta-d_r/4}\right)
\end{eqnarray}

and we prove in the appendix (via a Mellin transform) that 
\begin{equation} \label{bound1}
\sum_Hpq^H\exp\left(-p'(q'/2)^Hy\right)=\frac{p(p')^{-\delta}}{\log(2/q')}\Gamma(\delta)y^{-\delta}(1+o(1))
\end{equation} 
For $y\to\infty$ with $\delta=\frac{d_F-2}{d_r/2}$. We complete the proof by using  $y=n^{\theta-d_r/4}$. 
\end{proof}

Let us now study the second regime of $\theta$, $\theta<d_r/4$. 
\begin{theorem}\label{theo:without2}
For $\theta<d_r/4$, the proportion of covered mobile nodes tends to zero.
\label{theo-nc}\end{theorem}
For $u_i$ positive for all $i$:

let $P(u_0,u_1,\ldots,u_V,\ldots)=E\left[\prod_V u_V^{N_V(x)}\right]$, and 

let $P_V(u_0,\ldots,u_{V-1})=P(u_0,\ldots,u_{V-1},1,1,\ldots)$.
\begin{lemma}
For all integer $V$, $P_V(u_0,\ldots,u_{V-1})\ge 1-\frac{2^V}{\sqrt{n}}$. 
\end{lemma}
\begin{proof}
Naturally $P_V(u_0,u_1,\ldots,u_{V-1})\ge P_V(0,\ldots,0)$. The quantity $P_V(0,\ldots,0)$ is the probability that for all $i<V$, $N_i(x)=0$. This is equivalent to the fact that $x$ is always at distance larger than $R_n$ to any intersection of level higher than $V$. Since there are $2^V$ of such intersections, the measure of the union of all these intervals is smaller than $2^V/\sqrt{n}$. 
\end{proof}
\begin{lemma}
Assuming for all integer $i$, $u_i\le 1$, Let $V_n=\lceil \log_2\sqrt{n}\rceil$ we have $P(u_0,\ldots)\ge\prod_{V\ge V_n}u_i^{2^{V-V_n+1}}P_{V_n}(u_0,\ldots,u_{V_n-1})$.
\label{lem1}\end{lemma}
\begin{proof}
Indeed for $V\ge V_n$ we have $N_V(x)\le 2^{V-V_n+1}$, thus 
$$
P(u_0,\ldots)\ge\prod_{V\ge V_n}u_V 2^{V-V_n+1}E\left[u_0^{N_0(x)}\cdots u_{V_n-1}^{N_{V_n-1}(x)}\right].
$$
\end{proof}
\begin{lemma}
Let $V\le V_n$, the following holds:

$P_V(u_0,\ldots u_{V-1})\ge u_{V-1}P_{V-1}(u_0,\ldots,u_{V-2})$.
\label{lem2}\end{lemma}
\begin{proof}
For $V<V_n$ we have $N_V(x)\le 1$.
\end{proof}

\begin{proof}[Proof of theorem~\ref{theo-nc}]
The following identity holds:
\beq
I_n(H)=P\left(e^{-\rho p(H,0)},\ldots,e^{-\rho p(H,V)},\ldots\right).
\eeq
Clearly, with Lemma~\ref{lem1} we get
\begin{eqnarray*}
I_n(H)&\ge&\exp\left(-\rho\sum_{V\ge V_n}p(H,V) 2^{V-V_n+1}\right)\\
&&\times P_{V_n}\left(e^{-\rho p(H,0)},\ldots, e^{-\rho p(H,V_n-1)}\right).
\end{eqnarray*}
Since $\theta<d_r/4$, the first right hand factor tends to one as the quantity $\sum_{V\ge V_n}\rho p(H,V)=O\left(n^{\theta-d_r/4}p'(q'/2)^H\right)$ tends to zero. 

Using lemma~\ref{lem2} the second term is lower bounded for any integer $k<V_n$ by:
\beq
\prod_{i=1}^k\exp\left(-\rho p(H,V_n-i)\right)\left(1-2^{V_n-k}/\sqrt{n}\right)
\eeq
Each of the terms $\rho p(H,V_n-i)=O(n^{\theta-d_r/4}(q'/2)^{-i})$ tends to zero, thus $\exp(-\rho p(H,V-i))$ tends to 1. Since $2^{V_n}\le 2\sqrt{n}$ thus $1-2^{V_n-k}/\sqrt{n}\ge 1-2^{1-k}$. For any fixed $k$ we have $\lim\inf_n{I_n(H)}\ge 1-2^{1-k}$ uniformly in $H$. Since $k$ can be made as large as possible  thus $2^{2-k}$ as small as we want, the $I_n(H)$ to tend to 1 uniformly in $H$.
\end{proof}

\subsection{Extension of connectivity with drones}\label{connect_drones}

Now that we have analyzed the connectivity properties of the network and observed the situation when nodes are not covered. We have discovered two regimes, in function of the existing gNB infrastructure features: $\theta>d_r/4$ and $\theta<d_r/4$. Let us provide the necessary dimensioning of the network in terms of UAVs for ensuring the required services for the first case. As in the second case the number of isolated nodes is overwhelming, we consider drones are not a desirable, cost-effective solution for connectivity. We have, therefore, also identified the regime for which drones are to be deployed.  

\begin{theorem}[Connectivity with drones]\label{theo:with_drones}

Assume $n$ mobile nodes distributed according to a hyperfractal distribution of dimension $d_F$ and a gNB distribution of dimension $d_r$ and intensity $\rho$. If $\rho=n^\theta$ with $\theta>d_r/4$ then the average number of drones needed to cover the mobile nodes not covered by gNBs is
$D_n=nI_n\left(1+\sum_{k\ge 1}(k+1)^{2-d_F}\exp(-p'q'n^{\theta-dr/2}k^{d_r})\right)$.
\label{theodrone}\end{theorem}
\begin{remark} when $d_F>3$ and $\theta<d_r/2$ the distribution of the number of drones tends to be a power law of power $2-d_F$. When $\theta>d_r/2$ the average number of drones is asymptotically equivalent to the average number of isolated nodes. The probability that an isolated node needs more than one single drone decays exponentially. 
\end{remark}
\begin{lemma}
Let $I(H,R)$ be an upper bound of the probability that an interval of length $R$ on a road of level $H$ does not contain any gNB. We have
\beq
I(H,R)\le\exp\left(-\rho p'(q'/2)^HR^{d_r/2}\right).
\eeq
\end{lemma}
\begin{proof}
This is an adaptation of the estimate of $I_n(H)$ in the proof of theorem~\ref{theorelay}, where we replace $1/\sqrt{n}$ by $R$. Notice that $I_n(H)=I(H,R_n)$.
\end{proof}

Let us now look at the possibility of having gNBs at a Manhattan distance $R$. By Manhattan distance we consider the path from one mobile node to a gNB is 
\begin{itemize}
    \item either, the segment from the mobile node to the gNB if they are on the same road; in this case the distance is the length of this segment, and is  called a {\it one leg} distance;
    \item or, composed of the segment from the mobile node to the intersection to the road of the gNB and the segment from this intersection to the gNB; in this case the distance is the sum of the lengths of the two segments, and is called a {\it two-leg} distance.
\end{itemize}
\begin{lemma}
Let $J(R)$ be the probability that a mobile node has no gNB at two-leg distance less than $R$, $J(R)\le\exp\left(-\rho p'q' \frac{R^{d_r}}{1+d_r/2}\right)$. 
\end{lemma}
\begin{proof}
Let $J(H,R)$ be the probability that there is no relay on a road of level $H$ at a two-leg distance smaller than $R$. 
From the mobile nodes the maximal gap to the next intersection of level $H$ is $2^{-H+1}$, we have 
\begin{equation*}
J(H,R)\le\prod_{2^Hi\le R}I^2(H,2R-i2^{-H})
\end{equation*}
Each factor $I(H,2R-i2^{-H})$ comes from the fact that from the intersection of abscissa $i2^{1-H}$ the two segments apart of the perpendicular road of length $R-i2^{1-H}$ should not contain any relay. The power $2$ comes from the fact that we have to consider two intersections apart at distance $i2^{1-H}$ from the mobile node. We consider $H\ge H_R=\lceil\log_2 1/R\rceil$. For $H<H_R$ we will simply assume $J(H,R)\le 1$.
For $H\ge H_R$ we have 
\begin{equation*}
J(H,R)\le\exp\left(-2\rho p'(q'/2)^H\sum_{2^{-H} i\le R}(2R-i2^{-H})^{d_r/2}\right)
\end{equation*}
with the fact that 
\begin{eqnarray*}
\sum_{i=1}^{V_R}(2R-i2^{-V})^{d_r/2}&\ge& 2^H\int_0^{2R-1/2^H}x^{dr/2}dx\\
&=&2^H\frac{(2R-2^{-H})^{d_r/2+1}}{1+d_r/2}\\
&\ge&2^H\frac{R^{1+d_r/2}}{1+d_r/2}
\end{eqnarray*}
we obtain that
\begin{equation*}
J(H,R)\le\exp\left(-\rho p'(q')^H\frac{R^{1+d_r/2}}{1+d_r/2}\right)
\end{equation*}
The overall evaluation $J(R)$ is made of the product of all the $J(H,R)$ since the intersection and gNBs positions are independent:
\begin{eqnarray*}
J(R)&=&\prod_H J(H,R)\\
&\le&\exp\left(-\rho p'\sum_{H\ge H_R}(q')^H\frac{R^{1+d_r/2}}{1+d_r/2}  \right)\\
&=&\exp\left(-\rho p'(q')^{H_R}\frac{R^{1+d_r/2}}{1+d_r/2}  \right)
\end{eqnarray*}
With the estimate that $(q')^{H_R}\ge q' R^{d_r/2-1}$ we get $J(R)\le\exp\left(-\rho p'q'\frac{R^{d_r}}{1+d_r/2} \right)$.
\end{proof}
\begin{lemma}
Let $D(H,R)$ be the probability that for a mobile node on a road of level $H$  there is no gNB at Manhattan distance less than $R$. We have
$D(H,R)\le I(H,R)J(R)$. Let $k$ be an integer and $P_n(H,k)$ be the probability that a mobile node on road of level $H$ needs $k$ or more drones to be connected to the closest gNB. We have
\beq
P_n(H,k)\le I_n(H,kR_n)J((k-1)R_n)
\label{ineqP}\eeq
\end{lemma}
\begin{proof}
The fact that $D(H,R)\le I(H,R)J(R)$ comes from the fact that probability that there is no relay at distance $R$ is equal to  the product of the probabilities of the event: (i) there is no relay at one-leg distance smaller than $R$$(I(H,R)$), (ii) there is no relay at a two-leg distance smaller than $R$($J(R)$). 

The expression for $P_n(H,k)$ comes from the fact that to have $k$ or more drones we need no gNB within one leg Manhattan distance $kR_n$ and no gNB within two leg distance $(k-1)n$, for $k\ge 1$ since we have to lay an extra drone at road intersection. 

\end{proof}
\begin{proof}[Proof of Theorem~\ref{theodrone}]
The average number of drones needed to connect a mobile nodes on a road of level $H$ to the closest gNB is $L_n(H)=\sum_{k\ge 1} P(H,k)$.  
Thus $L_n(H)\le I_n(H)+\sum_{k\ge 1}I_n(H,kR_n) J(kR_n)$

Therefore the average total number $D_n$ of drones is given by:
\begin{eqnarray*}
    D_n&=&n\sum_{H}\lambda_H\sum_{k\ge 1}P(H,k)\\
    &&\le n\sum_H\lambda_H \sum_{k\ge 1}I_n(H,kR_n)J((k-1)R_n)\\
    &&\le nI_n+\sum_{k\ge 1} \sum_H nI_n(H,(k+1)R_n) \exp(-\rho p'q' k R_n^{d_r}k^{d_r})
\end{eqnarray*}
By an adaptation of~(\ref{bound1}) we have 
$\sum_H \lambda_H I_n(H,(k+1)R_n)\sim (k+1)^{-\delta d_r/2} I_n$ and by the fact that $\rho R_n^{d_r}=n^{\theta-d_r/2}$ we get the claimed result. 
\end{proof}

\begin{remark} The fact that the number of drones to connect the isolated mobile nodes is asymptotically equivalent to the number of isolated mobile nodes when $\theta>d_r/2$, is optimal when drone sharing is not allowed. We conjecture  that the isolated mobile nodes are so dispersed that sharing the connectivity of a drone is unlikely. 
In the other case when $d_r/4<\theta<d_r/2$ the number of drones per isolated node tends to be finite as soon as $d_F>3$.
\end{remark}


\section{Garages of drones}\label{garage}

While the fixed infrastructure is robust and can serve the vehicle requirements with a proper planning, the cost of installing and operating fixed infrastructure is substantial. In addition, as in many situations in telecommunications, the planning is often over-dimensionned in order to cope with "worst case scenarios'' (e.g., day versus night traffic).  In this sense, our scenario further offloads the traffic towards a flexible, ``ephemeral'' infrastructure, advancing towards the so-called ``moving networks of drones''. We will now propose a procedure for trimming the number of required relays, by converting some into garages of drones which will be dynamically deployed and launched to meet the requirements of the mobile users. The drones are to be seen as ``mobile relays'', with the possibility to build chains of drones that will serve the user with hop-by-hop communication through a highly reconfigurable Integrated Access and Backhaul (IAB).

\begin{figure}\centering
\includegraphics[scale=0.27]{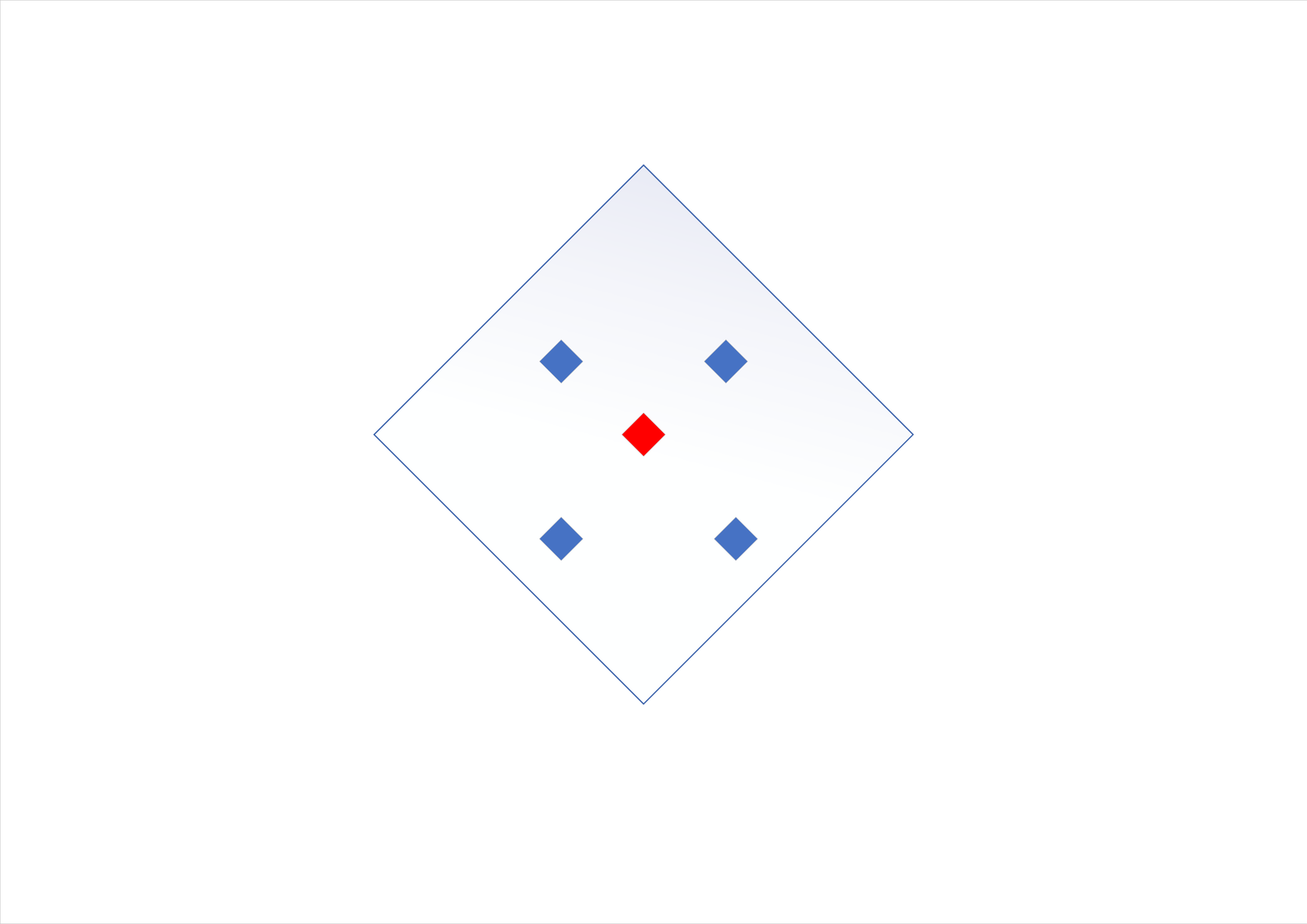}
\caption{Garage elimination algorithm}
\label{fig:garages_algorithm}
\end{figure}
Although it has been envisioned for the drones to be shared between operators, it is unlikely this will be done in an early phase: each operator will own and maintain its own fleet of drones.  Furthermore, as the UAVs are a cost effective and flexible solution for extending the coverage, the places for storing and charging them, what we call "garages",
are to be located in the same sites as the gNBs, owned by the operators.

We now provide a first straight-over insight on the home locations of the nomadic infrastructure. 

%

We define as the "{\it flight-to-coverage}" time, the time necessary for a drone to leave the garage and be in a distance lower than $R_n$ to the UE and to the gNB, such as to be able to form the backhaul. 
In order to fulfill the service requirements, the flight-to-coverage time should be lower than the allowed delay. 
Again, the drone needs not to be hoovering over the UE or the gNB but just have them in the range of its CSI-RS (Channel State Information Reference Signal) beams. 

We thus transform the flight-to-coverage time into constraints on connectivity (at all time), for a chain of drones of maximum length $k$. 
We want to select the relays that will be garages (and store drones) in order to minimize the number of garages and delay, i.e. under the constraint that at most $k$ hops are needed to forward the packets (and that the chains of drones are no longer than $k$).

We now describe informally an algorithm we use to select relays to become garages. 
Similar to a dominating set problem, we first make every relay a garage. We then check the garages in an arbitrary  order $A_1,A_2,\ldots$. We then eliminate sequentially a garage if it has four relays at Manhattan distance less than $R$, one in each of the four quadrants, e.g., as in~Figure \ref{fig:garages_algorithm}. We call these relays the "covering" relays. They have the property that every mobile node at Manhattan distance smaller than $R$ to the eliminated garage is necessary at Manhattan distance smaller than $R$ to at least one of the covering relay. We call this property the covering transfer property. 

Note that with the garage elimination heuristic, we eliminate the garage $A_\ell$ iff it has a covering set made of four relays of index smaller than $\ell$. 

\begin{lemma}
If a mobile node is at Manhattan distance smaller than $R$ to at least one relay, then it is at Manhattan distance smaller than $R$ to at least one non-eliminated garage.
\end{lemma}
\begin{proof}
Since there are no road with null density, the mobile node can be in any point of the map which is at Manhattan distance smaller than $R$ from a relay. Let us suppose that there is such point $X$ such that all relays at distance smaller than $R$ are eliminated garage. Let denote $i$ the smallest index among the indices of the relays at distance smaller than $R$. Since we suppose that its garage has been eliminated, the relay $A_i$ is covered by four relays of smaller index. Let $j<i$ be the index of the covering relay which is in the same quadrant of $A_i$ as the point $X$. Since $A_j$ and $X$ are at Manhattan distance smaller than $R$ of each other, this contradicts the fact that $i$ is the smallest index of the relays within distance $R$ from $X$. Therefore $X$ has necessarily at least one non-eliminated garage within distance $R$.


\end{proof}

Consequently, we will eliminate quickly all relays which have themselves a  neighboring relay in distance of less than $R$, therefore, there is no uncovered segment of road between the two relays.  
The remaining garages should hold drones to ensure connectivity within the delay tolerated. 
\begin{figure}\centering
\includegraphics[scale=0.27]{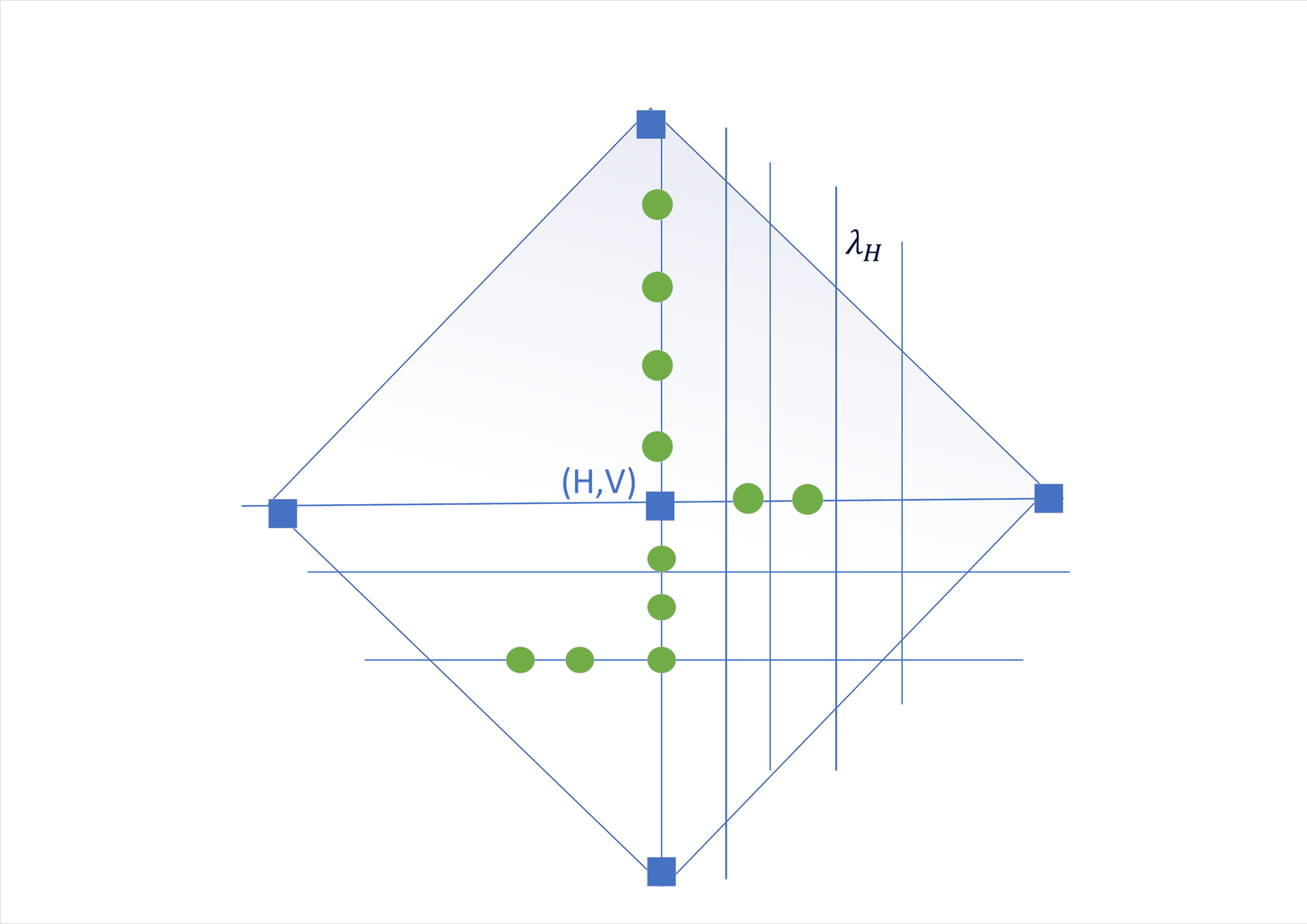}
\caption{Garage serving model}
\label{fig:garages}
\end{figure}

Figure~\ref{fig:garages} illustrates an operating model scenario with a garage of drones. A garage selected on a relay at an intersection of levels $(H,V)$ should be able to send its drones to all the cars arriving from the closest relays, forming together chains of drones of maximum length $k$. A chain of drones will be constructed with drones coming from both of the neighboring garages, as the passing of a car through a relay is announced to the neighboring relay through the wired F1 interface such that the following garage in the direction of movement can send its drones to meet the car to ensure the coverage along the path.  

A garage should be able to hold enough drones to ensure the continuity of service for all upcoming vehicles on the lines of the hyperfractal support crossing the area where the drones can move in an acceptable time. 
Note that as the arriving traffic flow in the "cell" served by the garage is a mixture of Poisson point processes on lines, the average incoming traffic the garage is serving can be computed therefore, given the distance towards the neighboring relays. 

We assume that the average speed of a drone is considered to be comparable to the maximum speed of a car.
The drones of a garage can thus be in three possible states: {\em in route} towards the car they will be serving, {\em serving} a car, and {\em coming back} from the service. 

Consequently, this procedure leads to the property that, within the "cell" around a garage, the graph is connected over time using the drones (as per Figure~\ref{fig:garages}). 
We will refer to this cell as a "moving network of drones".


\section{Numerical Evaluations}\label{simulations}

In this section, we first provide some simulations on the connectivity of gNBs, UEs and drones. We then provide some simulations on the garage locations and their properties.

The numerical evaluations are run in a locally built simulator in MatLab, following the description provided in Section~\ref{model} for both the stochastic modelling of the location of the entities and the communication model. 
The map length is of one unit and a scaling is performed in order to respect the scaling of real cities as well as communication parameters.

\subsection{Connectivity of gNBs, UEs and Drones.}

We first give some visual insight on the connectivity variation with the fractal dimension of the gNBs and $\theta$. Figure \ref{fig:dr3_1} shows in red $*$ the locations of vehicular UEs for $n=400$ and $d_F=3$ and in black circles the locations of the gNBs for $d_r=3$ and $\rho=n$. We are, therefore, in the first regime of $\theta$, $\theta>d_r/4$. On the other hand, Figure \ref{fig:dr3_2} displays, a snapshot of a network for the same $d_F$, $d_r$ and $n$, the second regime of $\theta$, with $\theta<d_r/4$, as more precisely $\theta=1/2$.

\begin{figure}  [ht]
\hspace*{2cm}
\begin{subfigure}[t]{0.225\textwidth}
\includegraphics[scale=0.45,trim={3cm 9.5cm 0cm 10cm}]{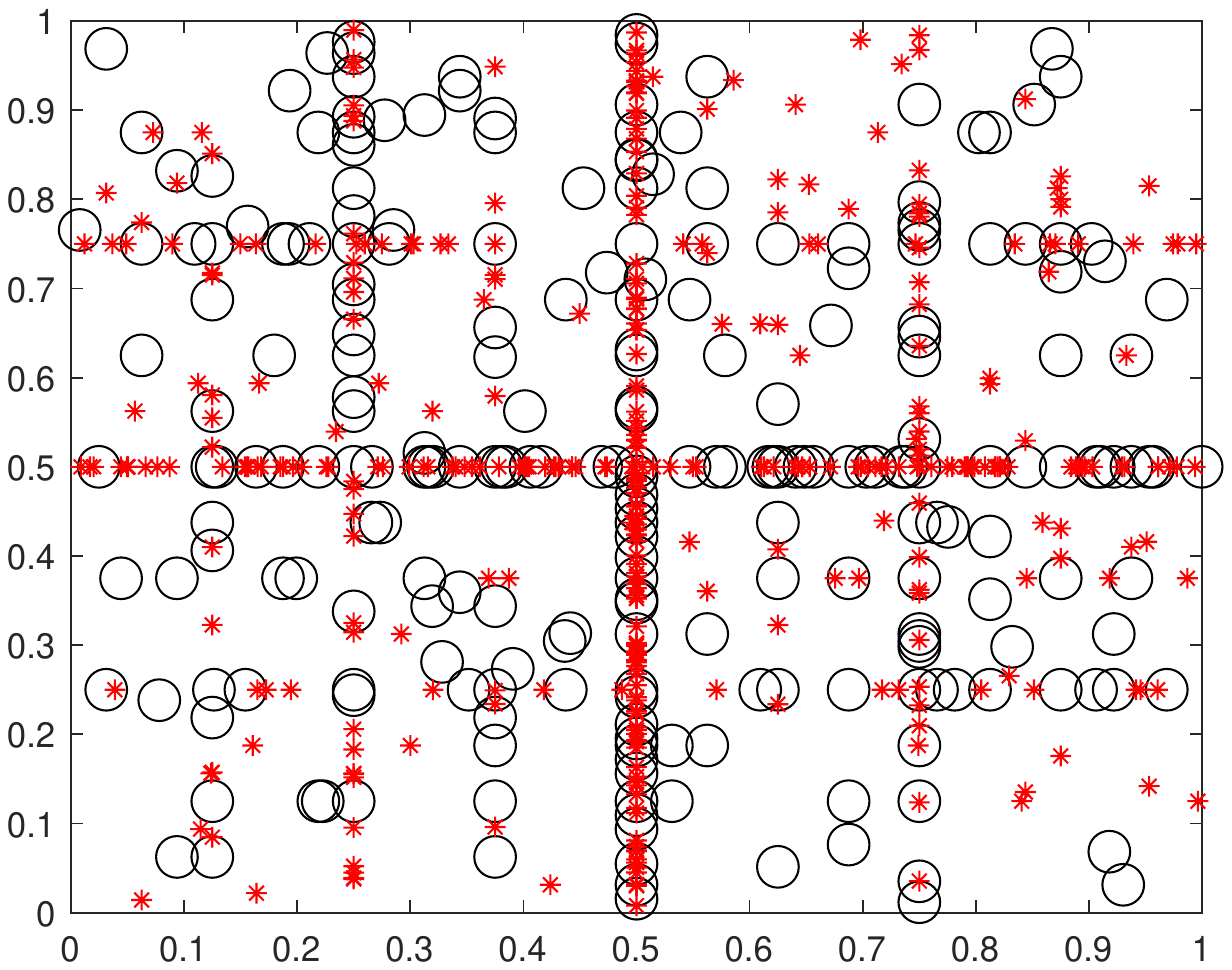}
\caption{$\theta>d_r/4$}
\label{fig:dr3_1}
\end{subfigure}
\hspace*{3cm}
\begin{subfigure}[t]{0.225\textwidth}
\includegraphics[scale=0.45,trim={3cm 9.5cm 0cm 10cm}] {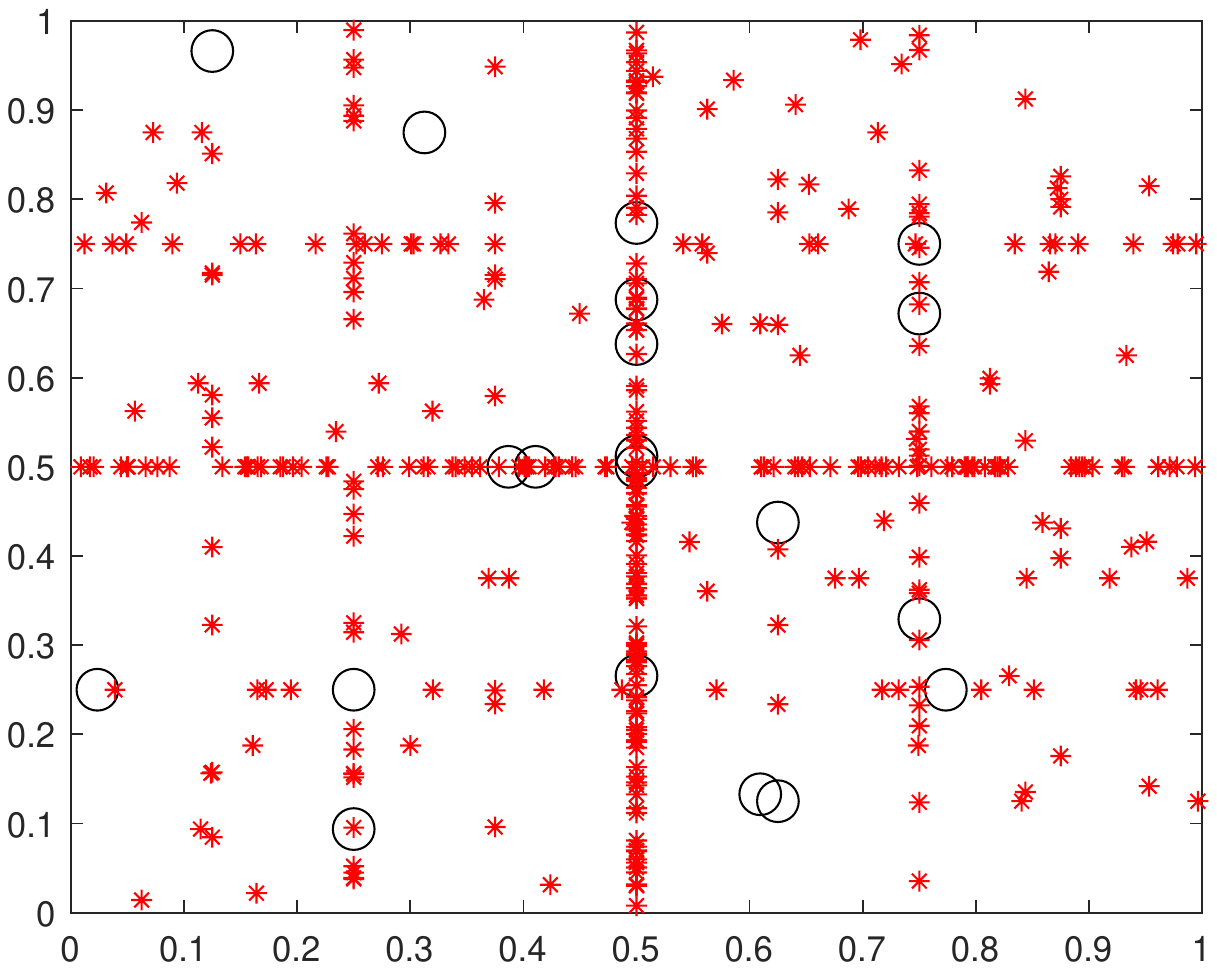}
\vspace*{-0.35cm}
\caption{$\theta<d_r/4$}
\label{fig:dr3_2}
\end{subfigure}
\caption{UEs and gNBs, $d_F=3$, $d_r=3$, $n=400$}
\label{fig:vx_gNBs1}
\end{figure}
Notice how  the number of gNB falls drastically for the second regime of $\theta$. This generates, as expected, and graphically visible in Figure  \ref{fig:vx_isolated1}, numerous disconnected UEs. For this regime, the number of isolated UEs is overwhelming, as clearly shown in Figure \ref{fig:dr3_2_iso}.

\begin{figure}  [ht]
\hspace*{2cm}
\begin{subfigure}[t]{0.225\textwidth}
\includegraphics[scale=0.45,trim={3cm 9.5cm 0cm 10cm}]{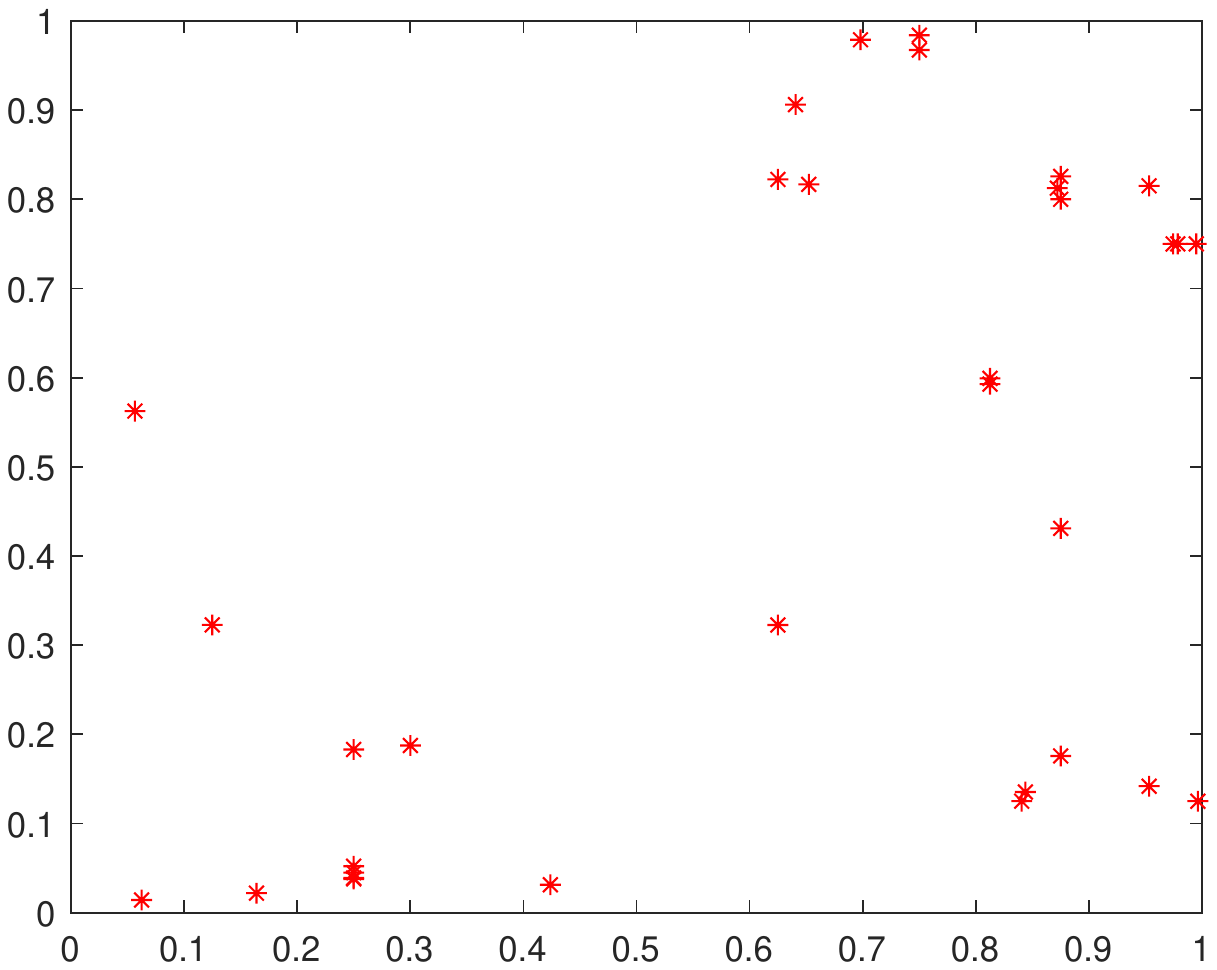}
\caption{$\theta>d_r/4$}
\label{fig:dr3_1_iso}
\end{subfigure}
\hspace*{3cm}
\begin{subfigure}[t]{0.225\textwidth}
\includegraphics[scale=0.45,trim={3cm 9.5cm 0cm 10cm}] {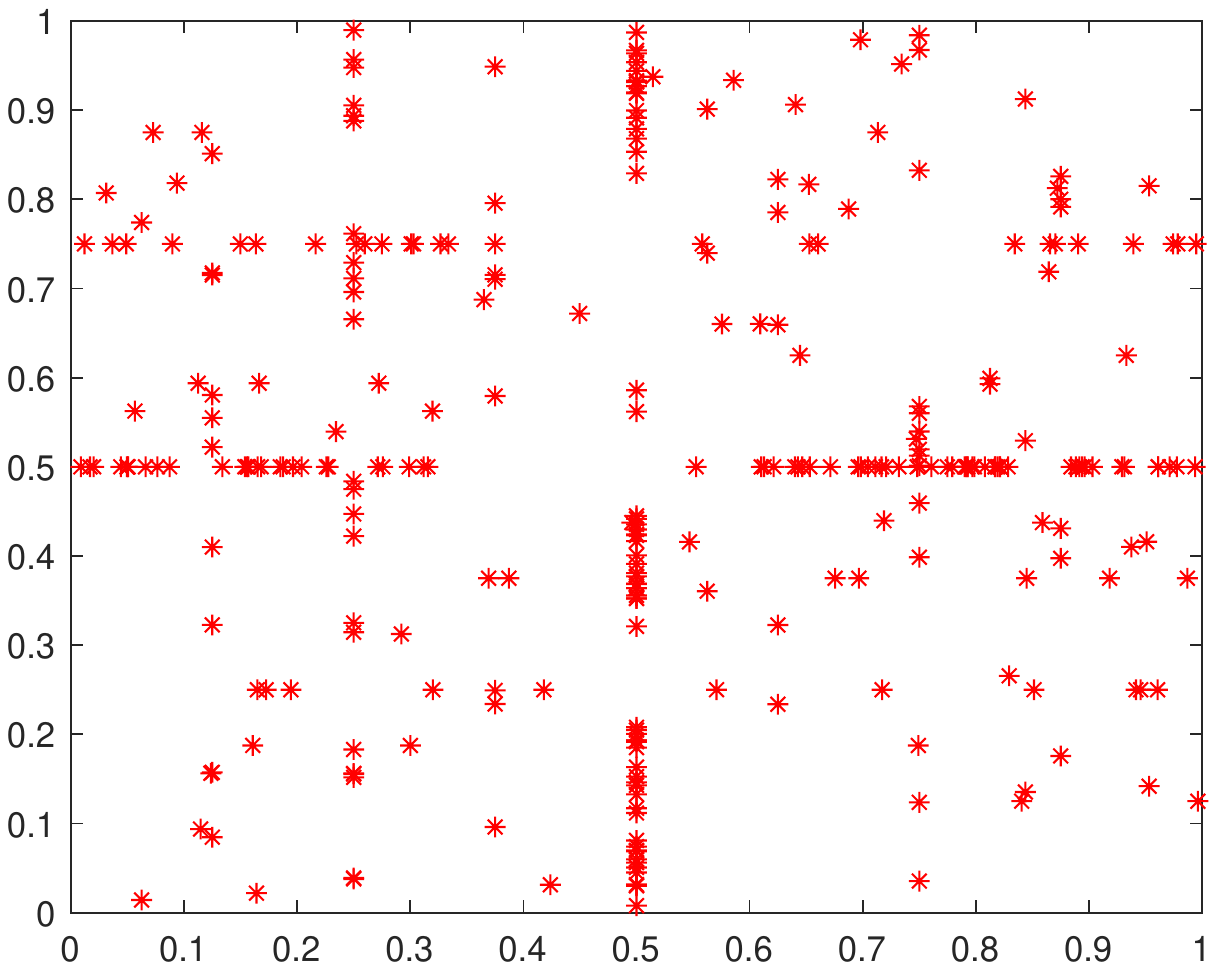}
\vspace*{-0.35cm}
\caption{$\theta<d_r/4$}
\label{fig:dr3_2_iso}
\end{subfigure}
\caption{Snapshot of UEs not covered, $d_F=3$, $d_r=3$}
\label{fig:vx_isolated1}
\end{figure}

We now look at what happens for a higher fractal dimension of the fixed telecommunication infrastructures. 
Figure \ref{fig:dr5_1} shows a snapshot of a network with $n=400$ vehicular UEs (in red $*$) and $d_F=3$ and gNBs with $d_r=5.5$ and $\rho=n$ (in black circles). This is here in the first regime of $\theta$, $\theta>d_r/4$ while Figure \ref{fig:dr5_2} displays, for the same $d_F$, $d_r$ and $n$, the second regime of $\theta$, $\theta<d_r/4$, more precisely, in this case, $\theta=1/2$.

\begin{figure}  [ht]
\hspace*{2cm}
\begin{subfigure}[t]{0.225\textwidth}
\includegraphics[scale=0.45,trim={3cm 9.5cm 0cm 10cm}]{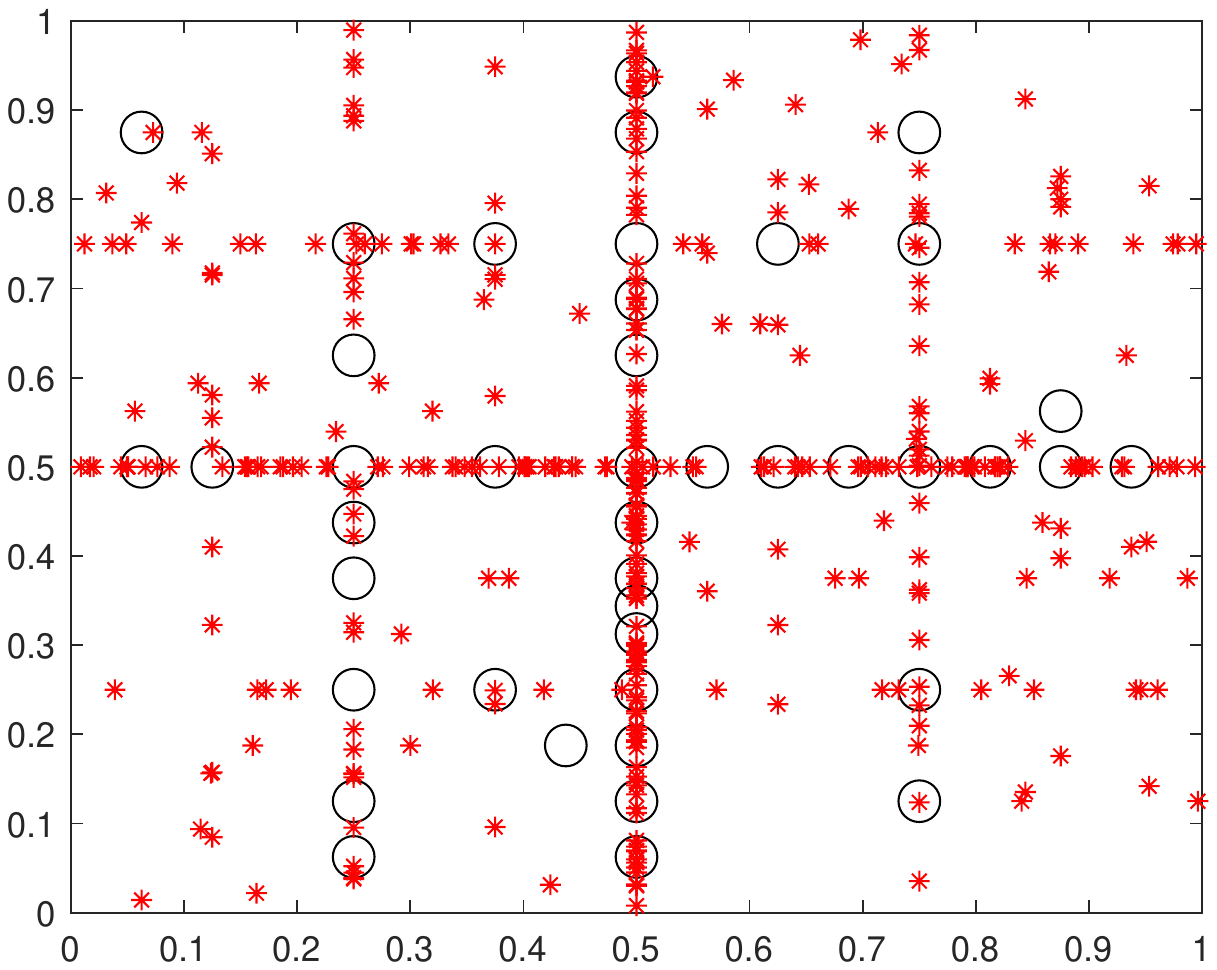}
\caption{$\theta>d_r/4$}
\label{fig:dr5_1}
\end{subfigure}
\hspace*{3cm}
\begin{subfigure}[t]{0.225\textwidth}
\includegraphics[scale=0.45,trim={3cm 9.5cm 0cm 10cm}] {_df_3_dr5_5_case_1}
\vspace*{-0.35cm}
\caption{$\theta<d_r/4$}
\label{fig:dr5_2}
\end{subfigure}
\caption{UEs and gNBs, $d_F=3$,  $d_r=5.5$, $n=400$ }
\label{fig:vx_gNBs2}
\end{figure}
Similarly to Figure \ref{fig:vx_isolated1}, Figure \ref{vx_isolated2} shows a snapshot of the isolated UEs for the two regimes of $\theta$ for $d_r=5.5$. Notice that the number of isolated nodes is significantly higher for a large fractal dimension of the eMBB infrastructure, even for the first regime of $\theta$. 

\begin{figure}  [ht]
\hspace*{2cm}
\begin{subfigure}[t]{0.225\textwidth}
\includegraphics[scale=0.45,trim={3cm 9.5cm 0cm 10cm}]{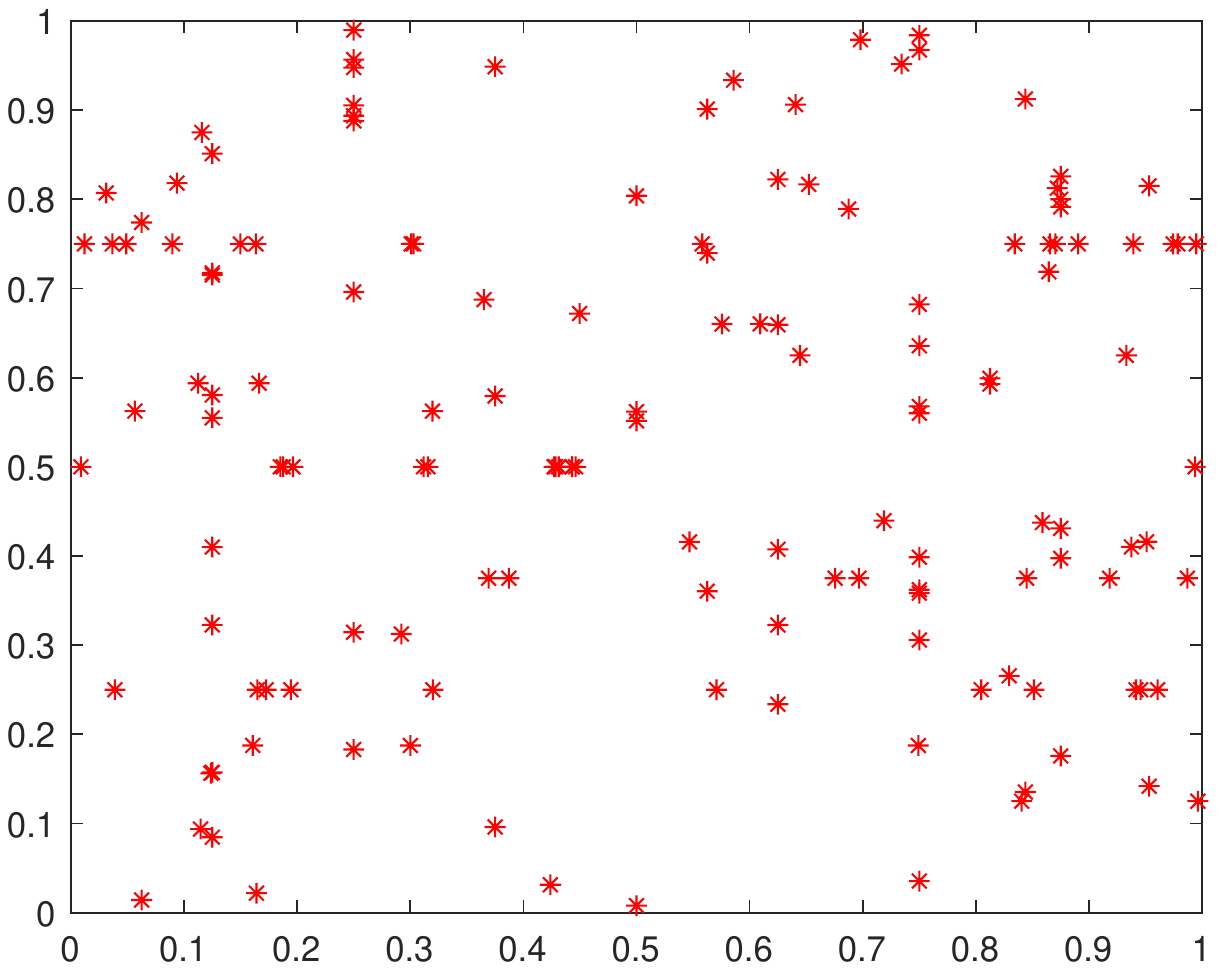}
\caption{$\theta>d_r/4$}
\label{fig:dr5_1_iso}
\end{subfigure}
\hspace*{3cm}
\begin{subfigure}[t]{0.225\textwidth}
\includegraphics[scale=0.45,trim={3cm 9.5cm 0cm 10cm}] {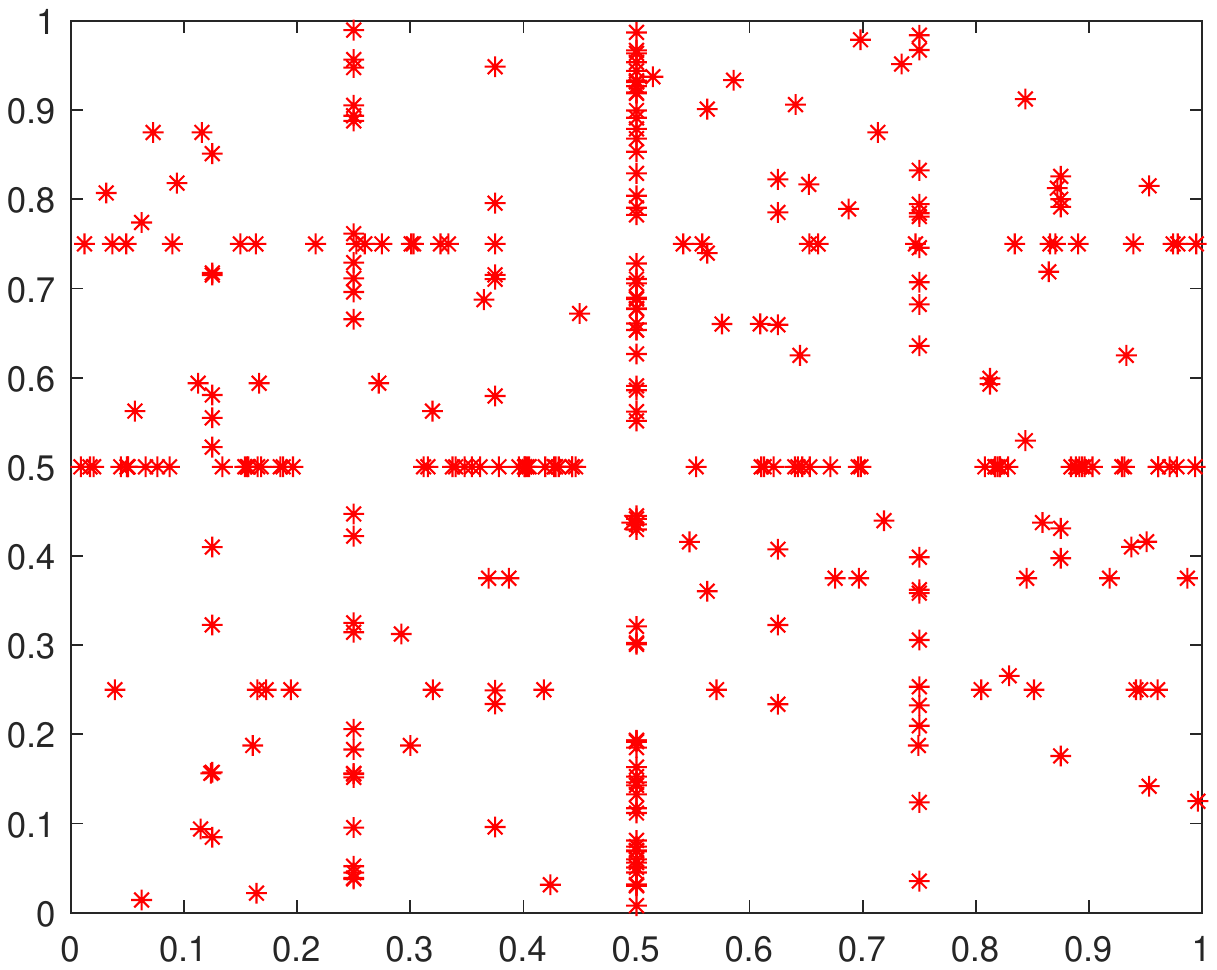}
\vspace*{-0.35cm}
\caption{$\theta<d_r/4$}
\label{fig:dr5_2_iso}
\end{subfigure}
\caption{Snapshot of UEs not covered, $d_F=3$, $d_r=5.5$}
\label{vx_isolated2}
\end{figure}

Let us now look at the validation of Theorem~\ref{without_drones} on the number of isolated nodes, this is an important parameter estimation as it gives the operator insight on the requirements for dimensioning the network. Figure~\ref{fig:without_drones1} shows the number of UEs that are not covered by a gNB when we vary the total number of devices and for two values of the fractal dimension of the gNBs: $d_r=3$ and $d_r=4$. In both cases, the fractal dimension of the nodes is $d_F=3$ and $n=\rho$, therefore $\theta=1$. 

\begin{figure}  [ht]
\hspace*{2cm}
\begin{subfigure}[t]{0.225\textwidth}
\includegraphics[scale=0.45,trim={3cm 9.5cm 0cm 10cm}]{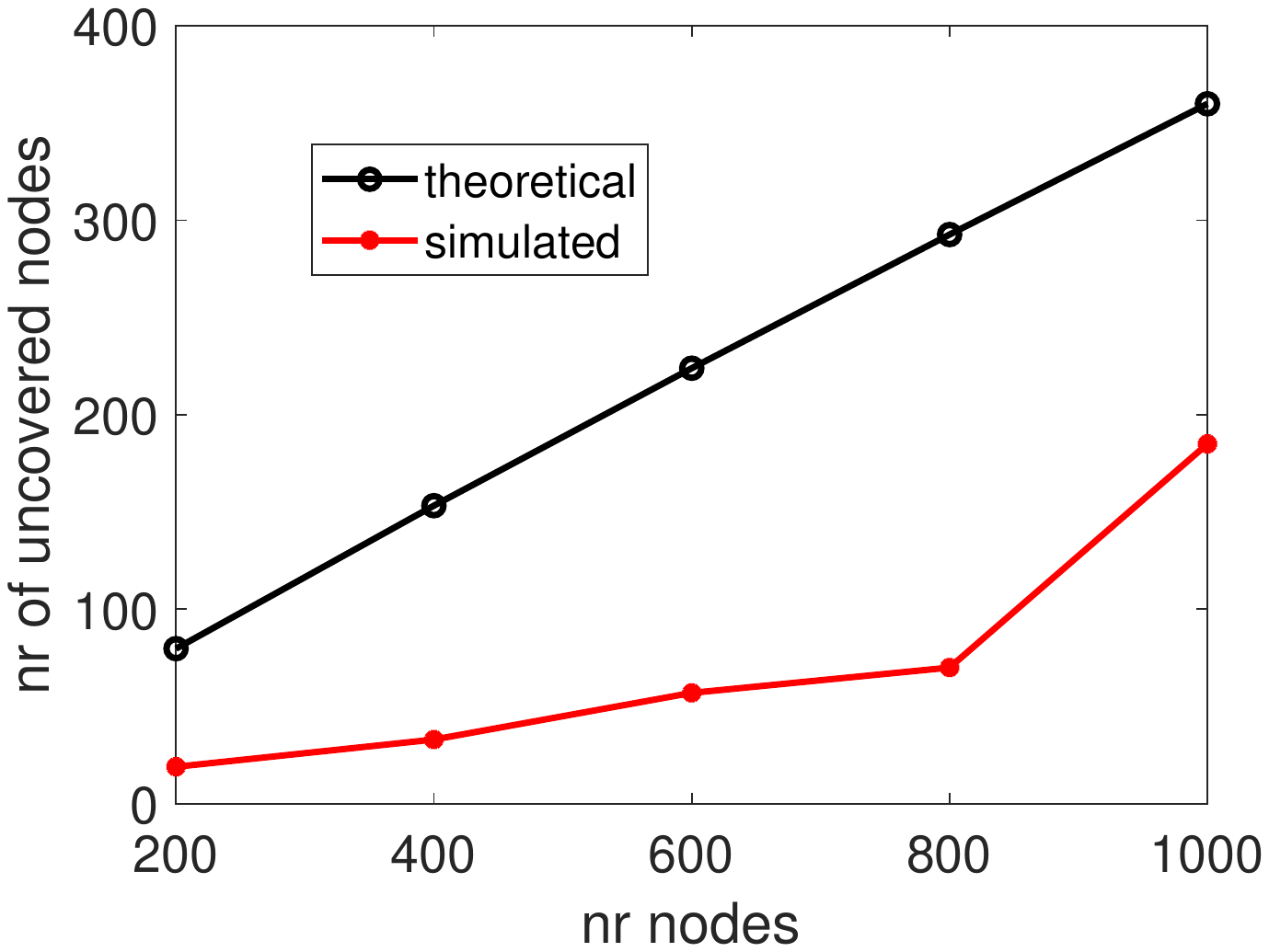}
\caption{}
\label{fig:dr3}
\end{subfigure}
\hspace*{3cm}
\begin{subfigure}[t]{0.225\textwidth}
\includegraphics[scale=0.45,trim={3cm 9.5cm 0cm 10cm}] {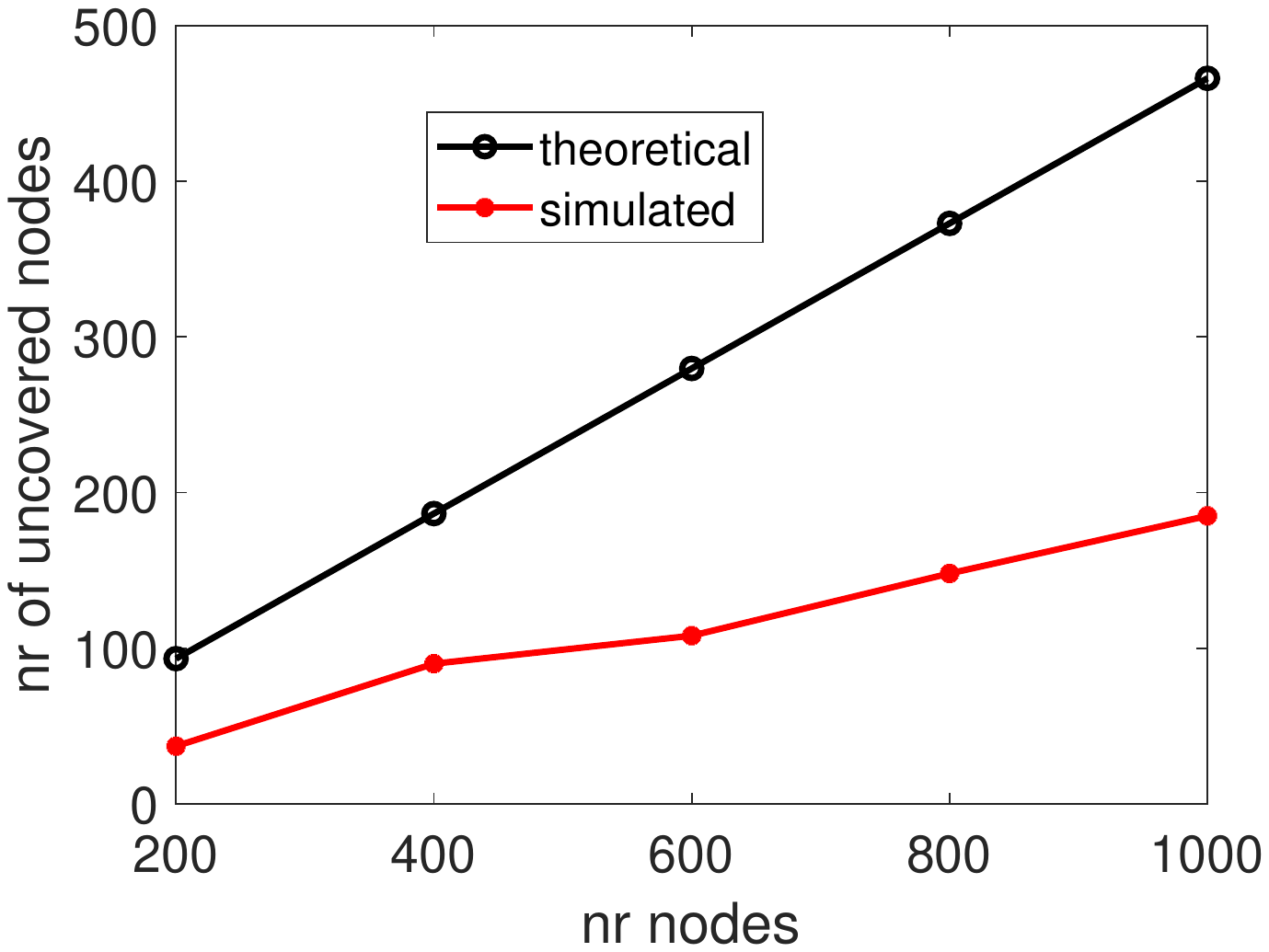}
\vspace*{-0.35cm}
\caption{}
\label{fig:dr4}
\end{subfigure}
\caption{(a) $d_r=3$, $d_F=3$; (b) $d_r=4$, $d_F=3$}
\label{fig:without_drones1}
\end{figure}
Notice that the bound we provided concurs with the simulation results. For the extension of this work and to provide an insight for the potential users of the tools in this work, we suggest using the expression in (\ref{tight}) for a tighter bound or the expression in (\ref{bound1}) of \ref{without_drones} if a close expression is desired. For instance, in this plot, we have used the latter.

Next, for three values of fractal dimension of gNBS, $d_r=3$, $d_r=4$ and $d_r=5$ respectively, yet for a case of $\theta=1/2$, we show in Figure \ref{fig:without_drones2}, that the number of isolated nodes tends to the actual number of nodes in the network, as stated in Theorem \ref{theo:without2}.

\begin{figure}[ht]\centering
\includegraphics[scale=0.45,trim={0cm 9.5cm 0cm 10cm}]{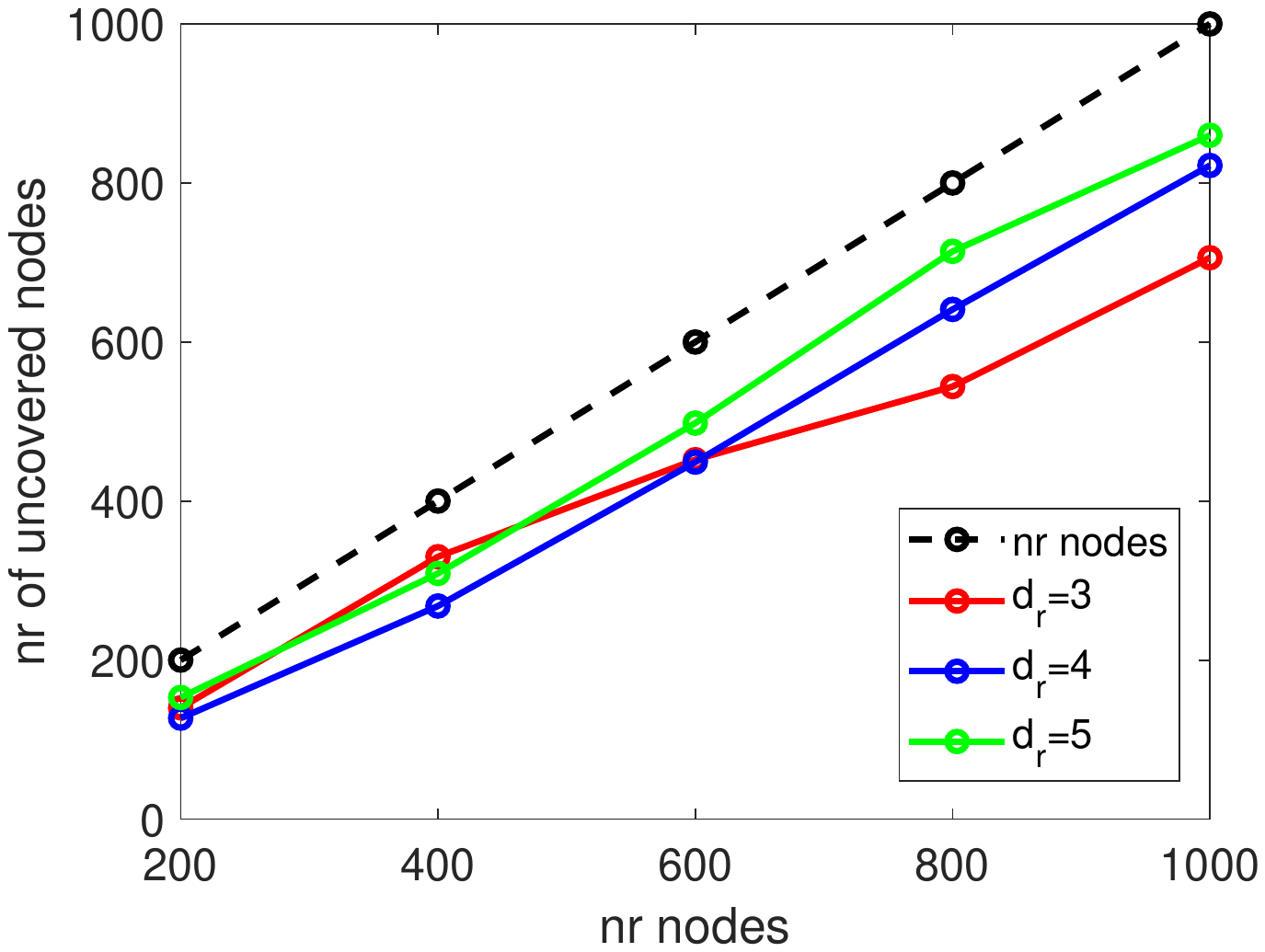}
\caption{Proportion of isolated nodes for $\theta<d_r/4$}
\label{fig:without_drones2}
\end{figure}

This confirms again that, in the case when $\theta<d_r/4$, the eMBB infrastructure alone cannot provide the required connectivity and consequently the number of disconnected nodes is overwhelming.

Figure \ref{with_drones1} illustrates the result proved in Theorem \ref{theo:with_drones}: the number of drones required to ensure connectivity for the isolated UEs (when $\theta>d_r/4$) behaves asymptotically like the number of isolated nodes.  

\begin{figure}
\hspace*{2cm}
\begin{subfigure}[t]{0.225\textwidth}
\includegraphics[scale=0.45,trim={3cm 9.5cm 0cm 10cm}]{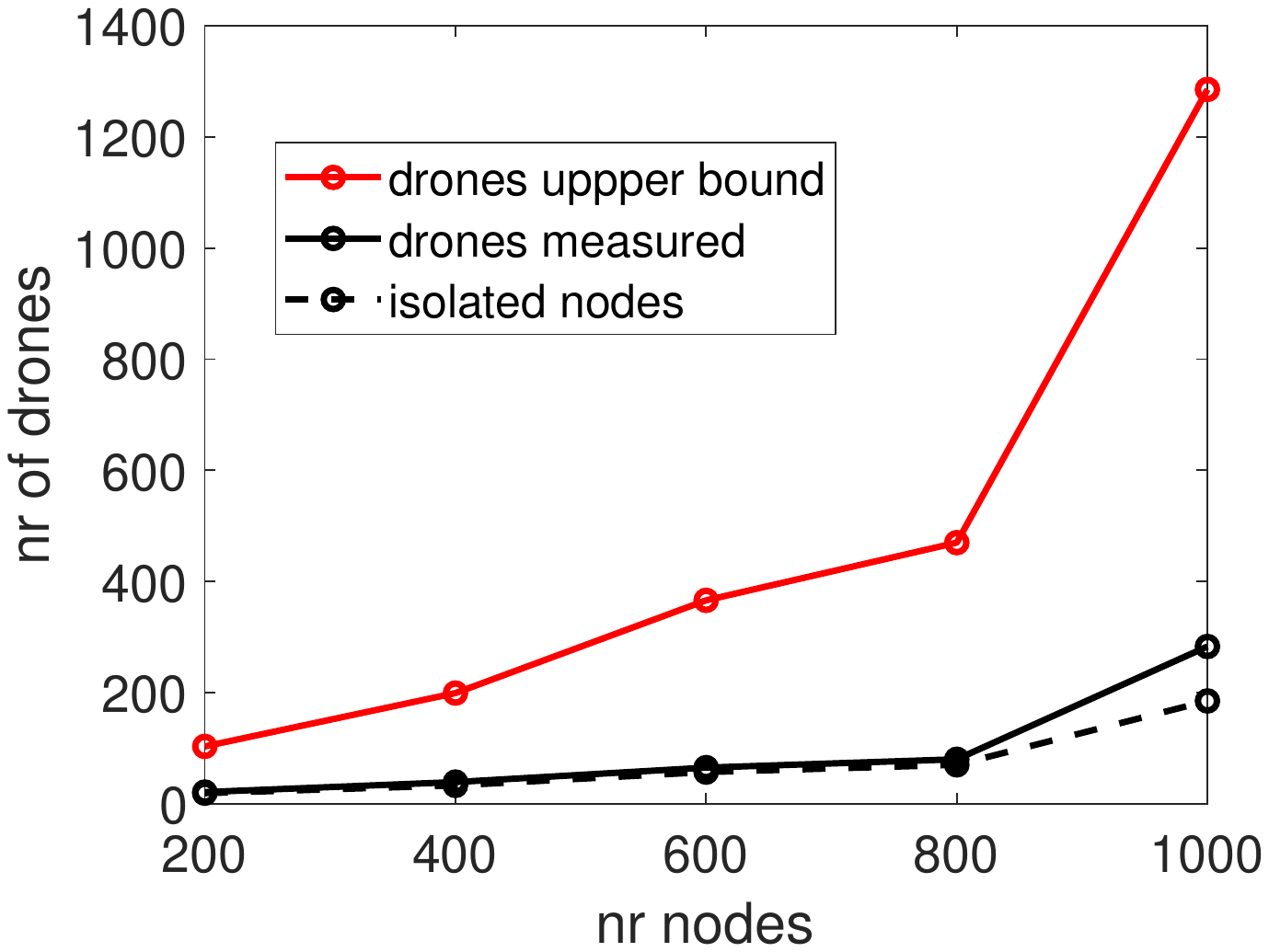}
\caption{$d_r=3$, $d_F=3$;}
\end{subfigure}
\hspace*{3cm}
\begin{subfigure}[t]{0.225\textwidth}
\includegraphics[scale=0.45,trim={3cm 9.5cm 0cm 10cm}] {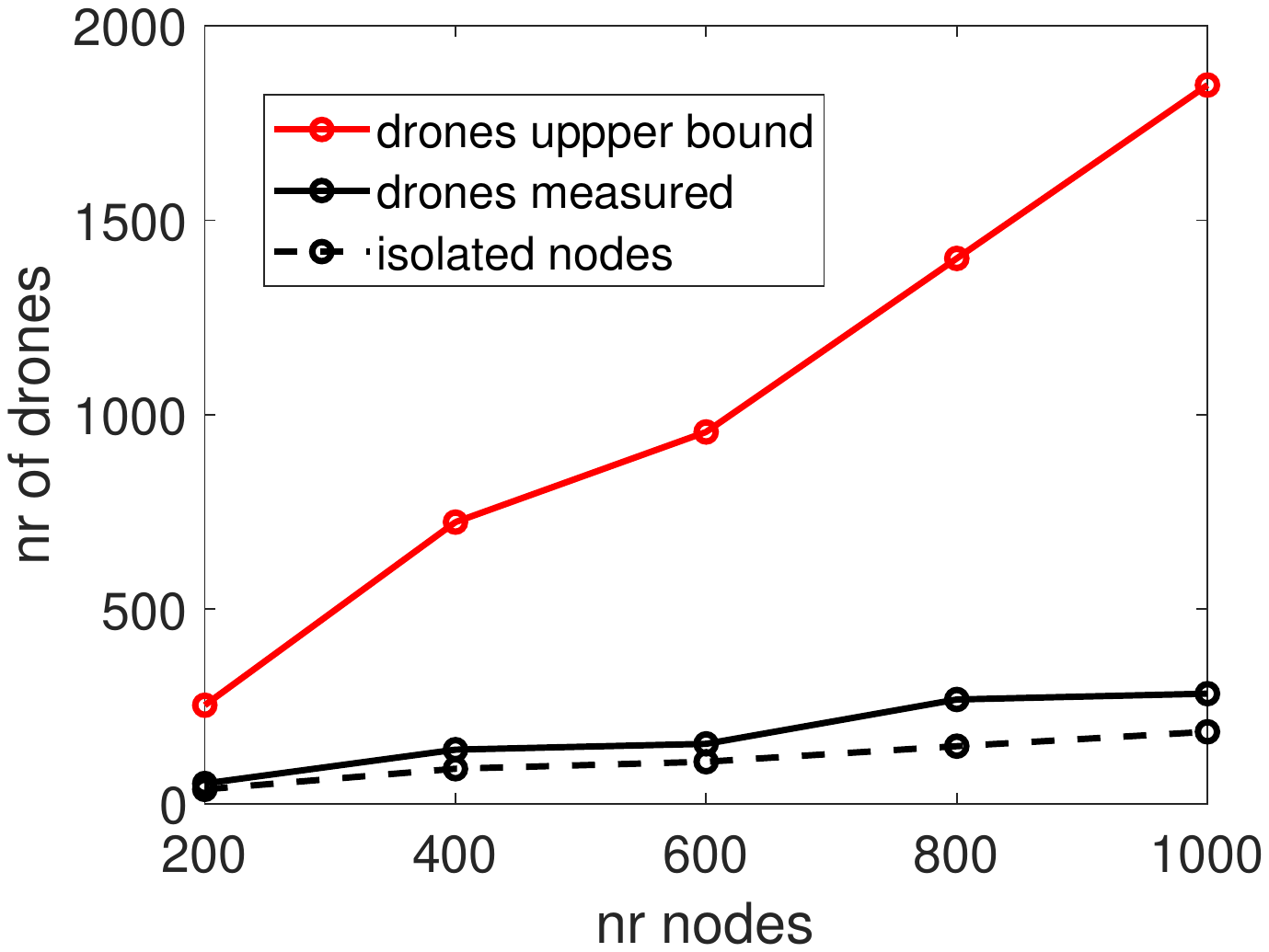}
\vspace*{-0.35cm}
\caption{$d_r=4$, $d_F=3$}
\end{subfigure}
\caption{Number of drones for connectivity }
\label{with_drones1}
\end{figure}

\subsection{Location and Size of Drone Garages.}

We run some simulations in order to validate the physical distance requirements between mobile nodes, relays and garages. In order to have realistic figures we have assumed that the practical range of emission with mmWave is of order 100m, and that the density of mobile nodes is around 1,000 per square kilometer. This would lead to $R_n=\sqrt{10/n}$, equivalent to a square of 1 square kilometer has a side length of 10 times the radio range and contains 1,000 mobile nodes. We fix $d_F=3$ and $d_r\approx 2.3$ (with $p_r=0.1$).

We first compute the distribution of the distance of the mobile nodes to their closest base stations expressed in hop count in~Figure~\ref{fig_coverageRelay10}. A hop count of one means that the mobile nodes is in direct range to a base station. A hop count of $k$ ($k$ integer) means that the mobile nodes would need $k-1$ drones to let it connected to its closest base station. We display the distribution for various values of $n$ (green: $n=5,000$; blue: $n=10,000$; red: $n=20,000$; brown: $n=40,000$; black $n=80,000$).

\begin{figure}
\vspace{3cm}
\begin{subfigure}[t]{0.225\textwidth}
\includegraphics[scale=0.4,trim={3cm 9.5cm 0cm 10cm}]{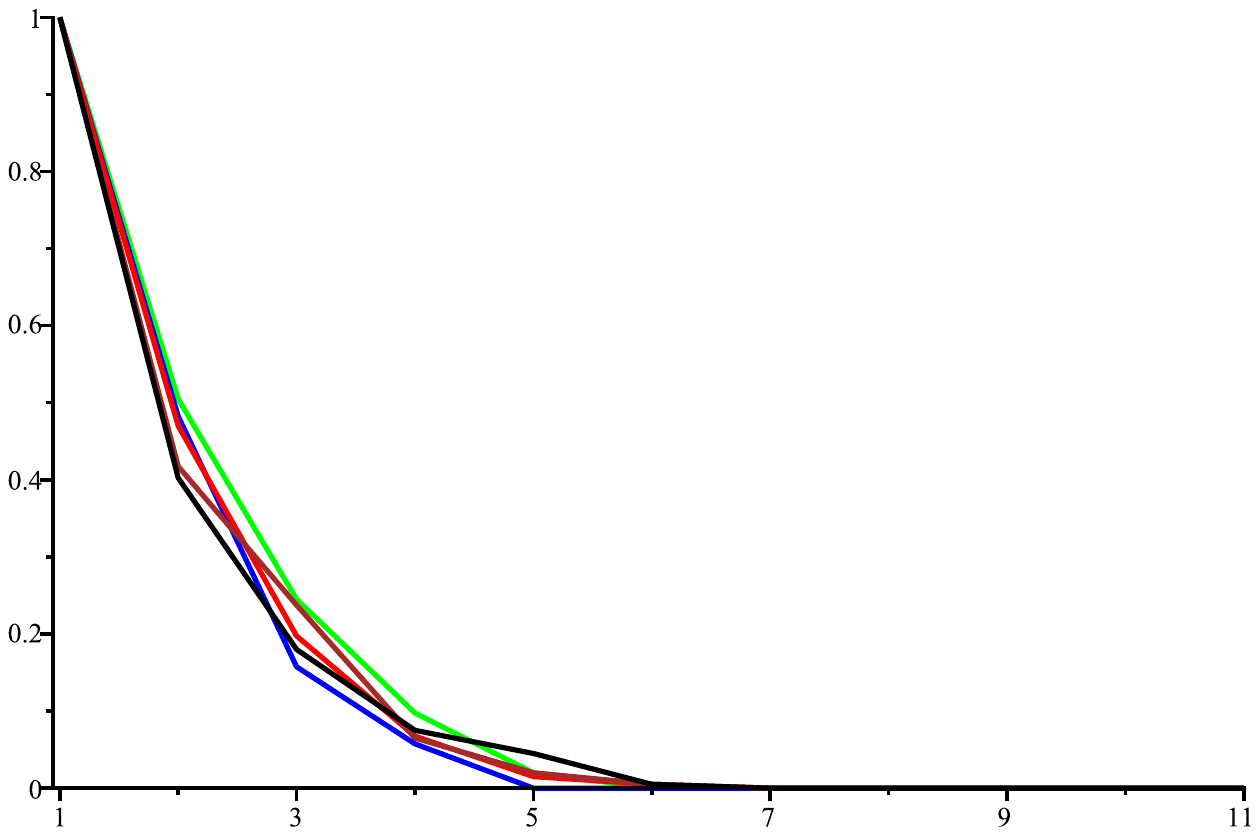}
\vspace*{-3cm}
\caption{$\theta=d_r/4$;}
\end{subfigure}
\hspace*{2cm}
\begin{subfigure}[t]{0.225\textwidth}
\includegraphics[scale=0.4,trim={3cm 9.5cm 0cm 10cm}] {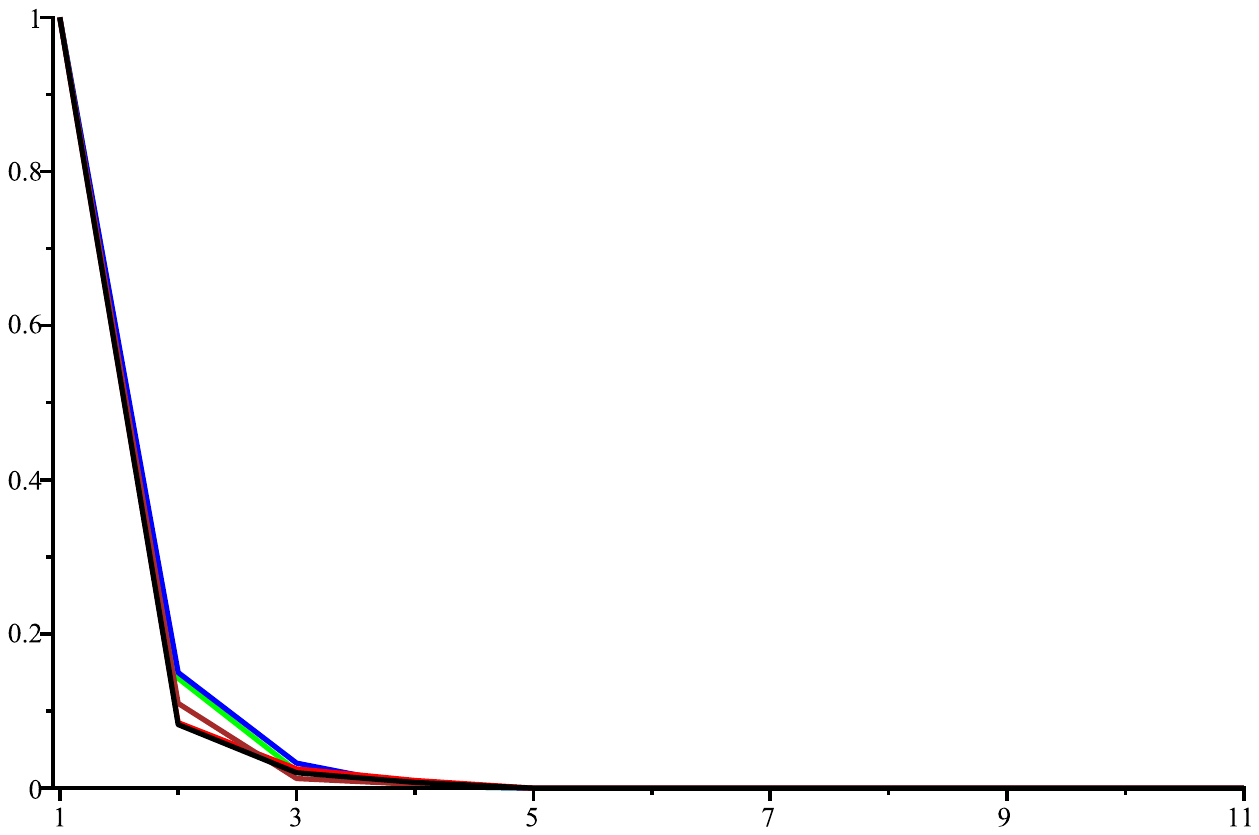}
\vspace*{-3cm}
\caption{$\theta=1.2d_r/4$;}
\end{subfigure}
\hspace*{2cm}
\begin{subfigure}[t]{0.225\textwidth}
\includegraphics[scale=0.4,trim={3cm 9.5cm 0cm 10cm}] {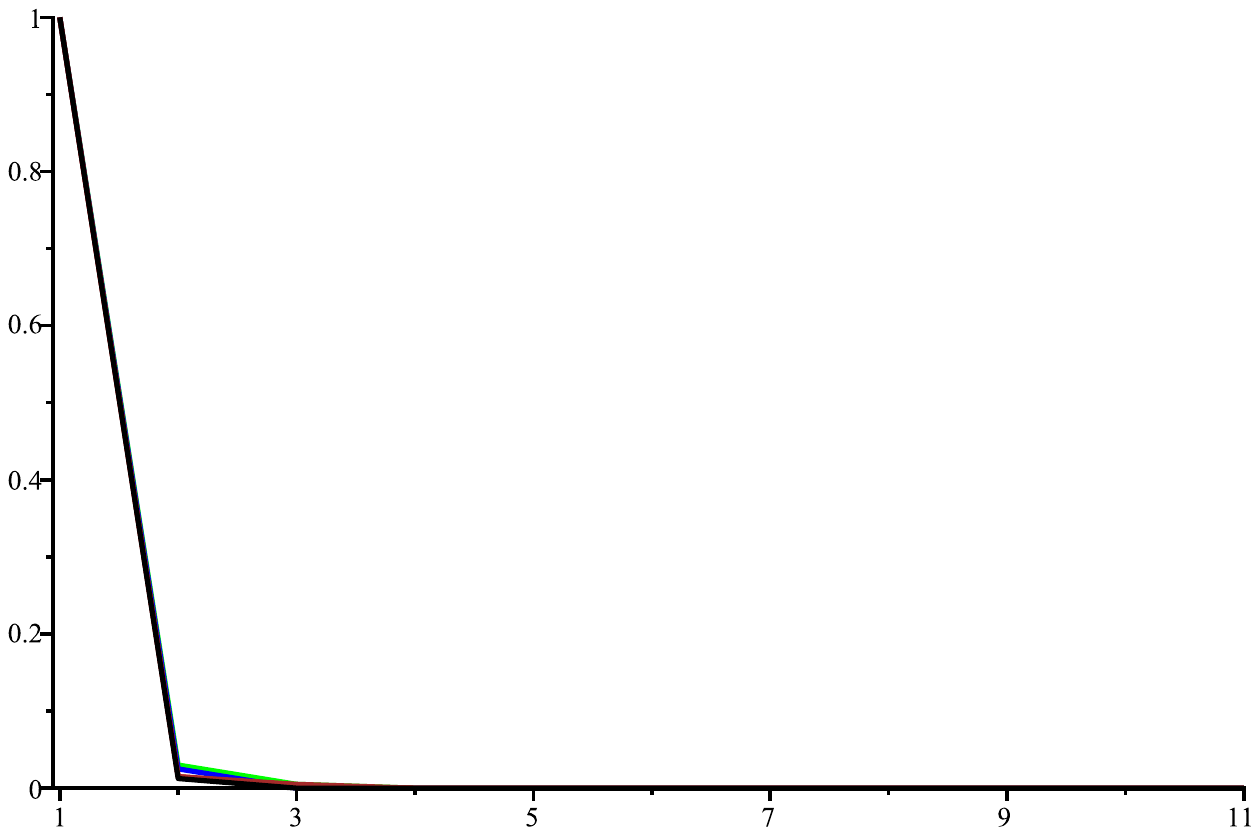}
\vspace*{-3cm}
\caption{$\theta=1.4d_r/4$;}
\end{subfigure}
\vspace{-2cm}
\caption{distribution of distance to base stations in hop count}
\label{fig_coverageRelay10}
\end{figure}

Secondly, in~Figure~\ref{fig_meanDrones10}, we compute the average number of isolated nodes (in blue) ({\it i.e.} the mobile nodes not at a direct range to a base station), and at the same time the average number of drones (in brown) to connect them to the closest base station. The two numbers are given as a fraction of the total number of mobile nodes present in the map. 

\begin{figure}
\vspace{3cm}
\begin{subfigure}[t]{0.225\textwidth}
\includegraphics[scale=0.4,trim={3cm 9.5cm 0cm 10cm}]{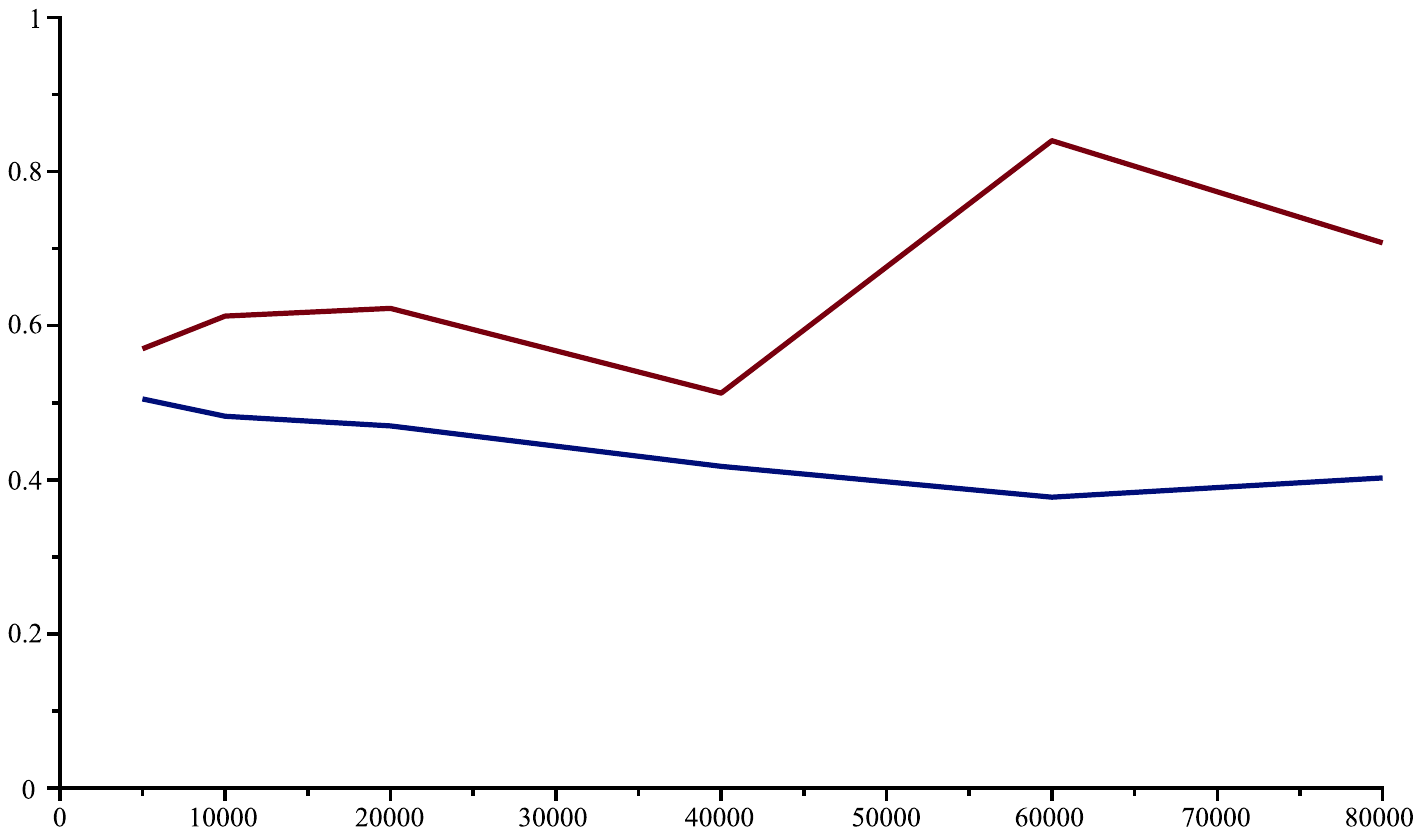}
\vspace*{-3cm}
\caption{$\theta=d_r/4$;}
\end{subfigure}
\hspace*{2cm}
\begin{subfigure}[t]{0.225\textwidth}
\includegraphics[scale=0.4,trim={3cm 9.5cm 0cm 10cm}] {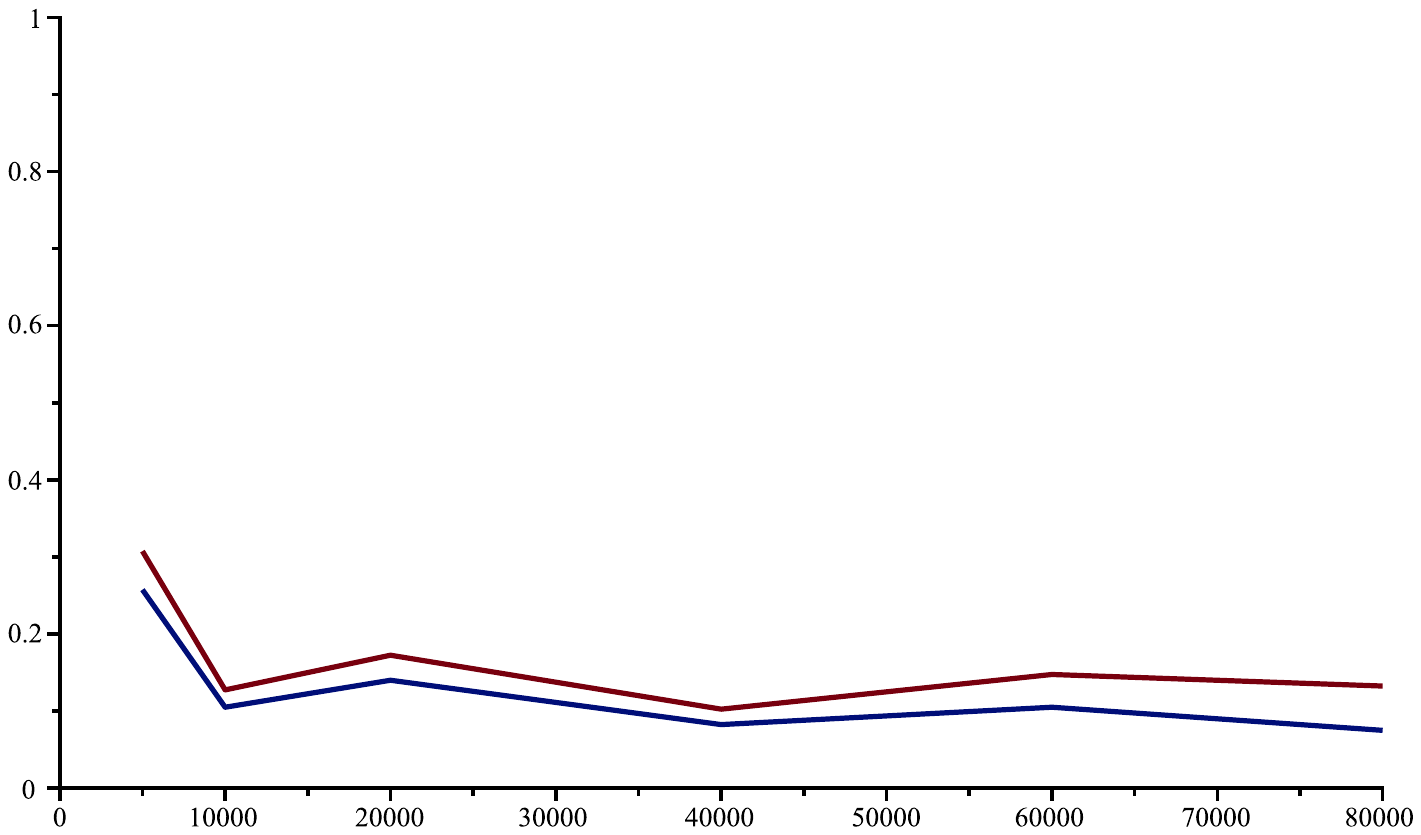}
\vspace*{-3cm}
\caption{$\theta=1.2d_r/4$;}
\end{subfigure}
\hspace*{2cm}
\begin{subfigure}[t]{0.225\textwidth}
\includegraphics[scale=0.4,trim={3cm 9.5cm 0cm 10cm}] {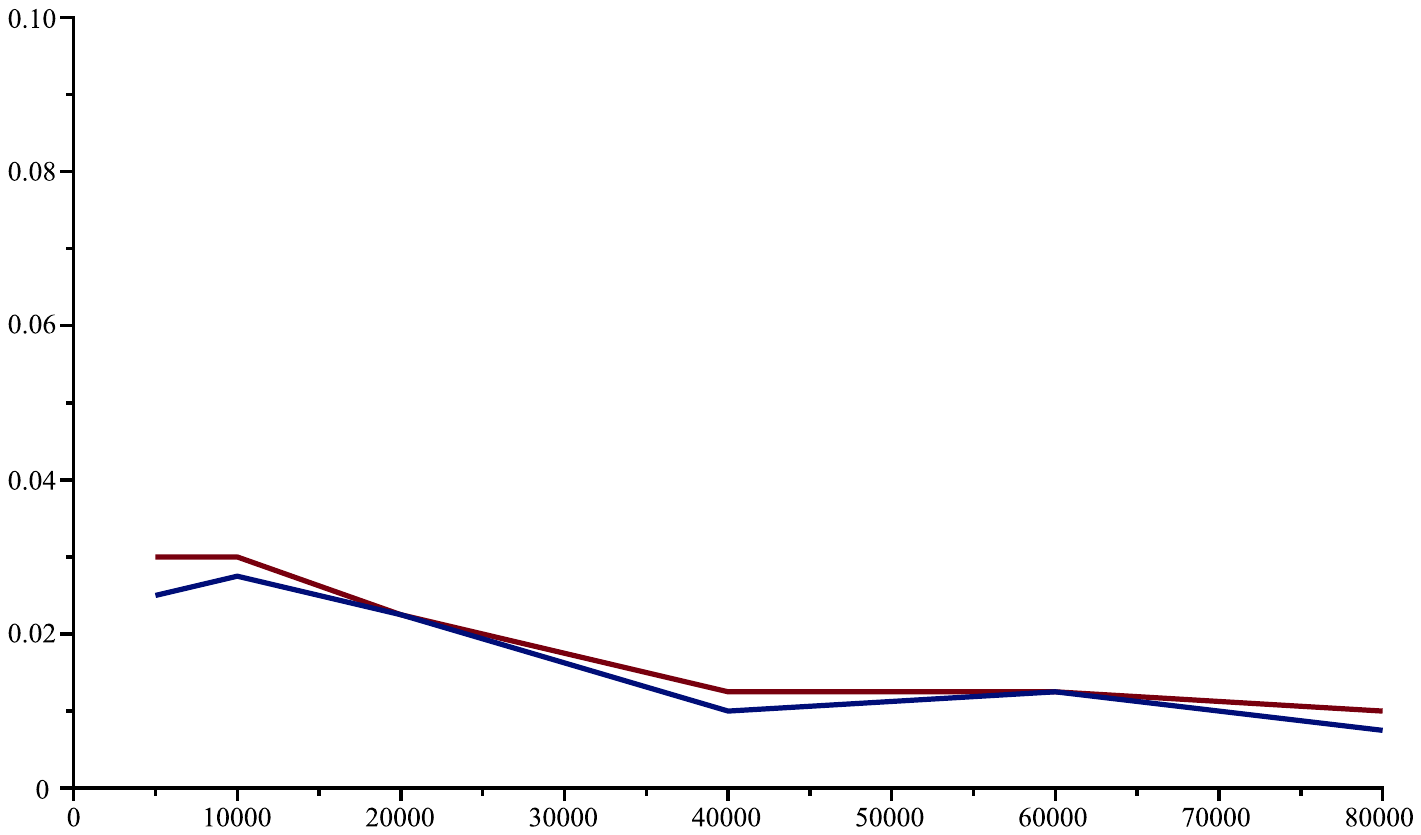}
\vspace*{-3cm}
\caption{$\theta=1.4d_r/4$ (notice the change of scale);}
\end{subfigure}
\vspace{-1cm}
\caption{proportion of isolated nodes (blue), proportion of drones (brown), versus the total number $n$ of mobile nodes}
\label{fig_meanDrones10}
\end{figure}

In Figure~\ref{fig_MapMPR} we display the map of the garage locations after reduction on base stations. The reduction is made according to various coverage radii (the parameter $R$ respectively) which varies from radio ranges 5 (500m) to 80 (8km). We can see as expected that the garage density decreases when the coverage radius increases. The parameters are $n=50,000$ and $\theta=1.2\frac{d_r}{4}$. The coverage radius impacts the delay at which the drones can move to new mobile nodes, although many of these moves could be easily predicted from the aim and trajectory of the mobile nodes. Figure~\ref{fig:MPRSize} shows the variation of size of the garage set as function of the coverage radius. The parameters are $n=50,000$; in brown $\theta=d_r/4$; in blue $\theta=1.2 d_r/4$; and in green $\theta=1.4 d_r/4$. When the coverage radius is zero, every relay is a garage and we get the initial number of relays. We notice that when the coverage radius tends to infinity the limit density of garages is not bounded and increases with the number of relays. We conjecture that it increases as the logarithm of this number. On the right subfigure, we display the size of the garage set when the city map is considered on a torus without border. In this case the garage size decreases to 1 when the coverage radius increases. Figure~\ref{fig_MapMPRtorus} gives examples of garage maps in a torus.

\begin{figure}
\vspace{5cm}
\begin{subfigure}[t]{0.225\textwidth}
\includegraphics[scale=0.55,trim={3cm 9.5cm 0cm 10cm}]{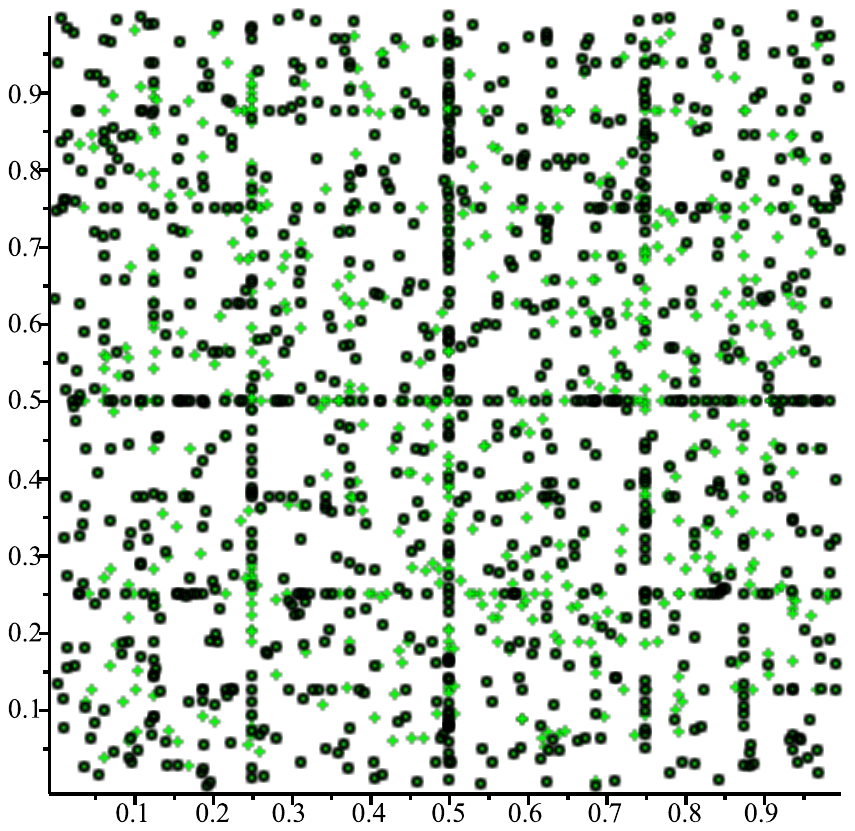}
\vspace*{-4cm}
\caption{coverage distance $5R_n$}
\end{subfigure}
\hspace*{2cm}
\begin{subfigure}[t]{0.225\textwidth}
\includegraphics[scale=0.55,trim={3cm 9.5cm 0cm 10cm}] {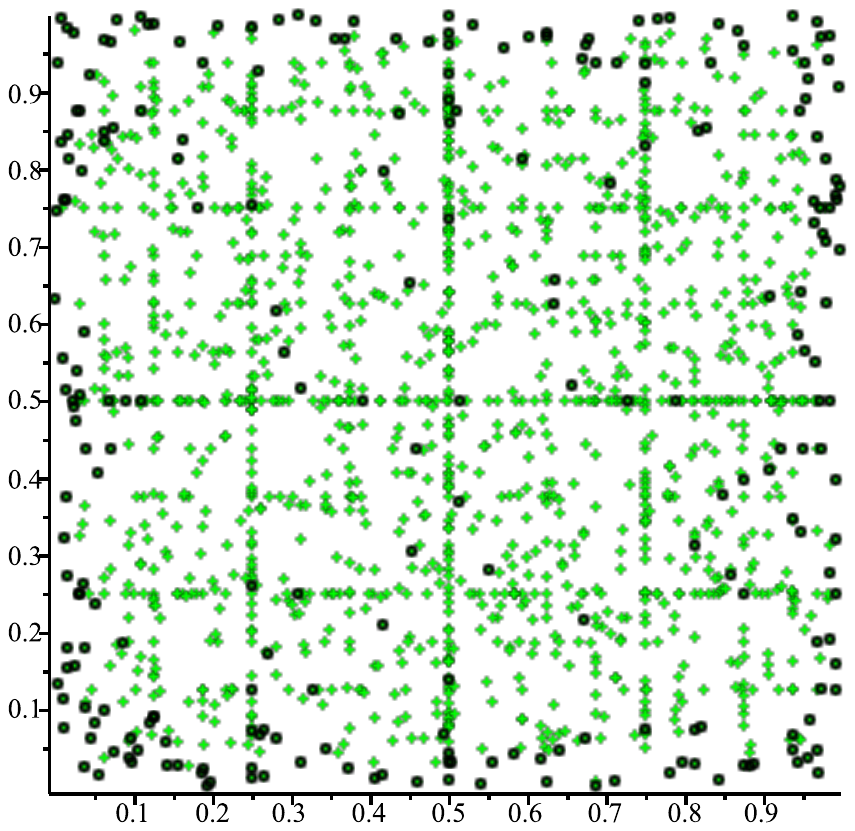}
\vspace*{-4cm}
\caption{coverage distance $20R_n$}
\end{subfigure}
\hspace*{2cm}
\begin{subfigure}[t]{0.225\textwidth}
\includegraphics[scale=0.55,trim={3cm 9.5cm 0cm 10cm}] {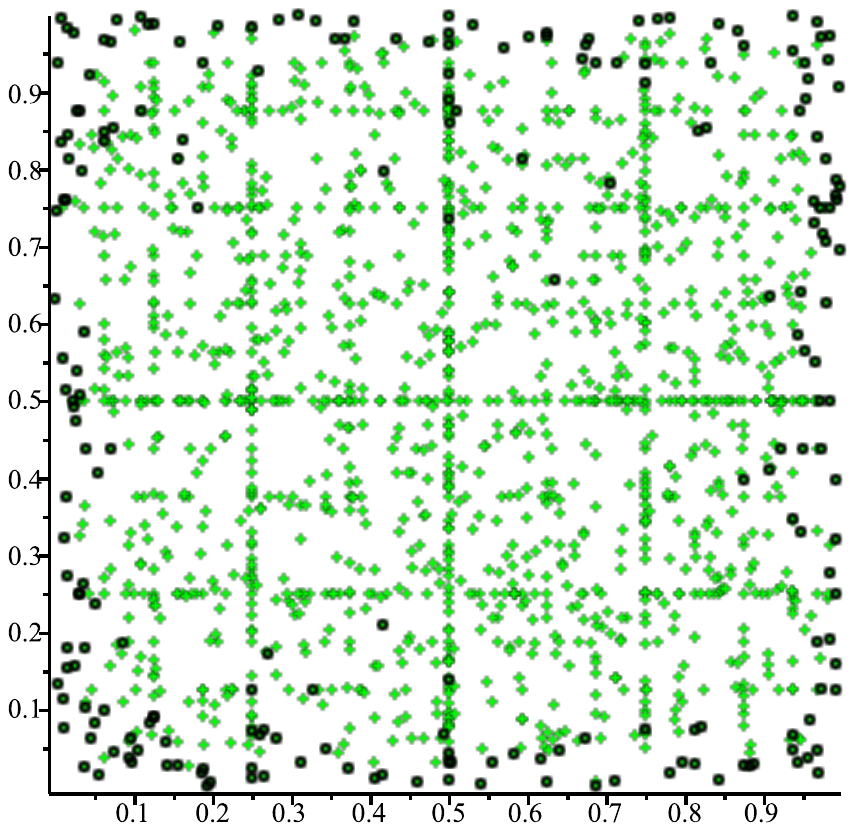}
\vspace*{-4cm}
\caption{coverage distance $80R_n$}
\end{subfigure}
\vspace{-2cm}
\caption{Map of drone garages (black circles), among the base stations (green crosses), for different coverage distances.}
\label{fig_MapMPR}
\end{figure}

\begin{figure}
\vspace{5cm}
\begin{subfigure}[t]{0.425\textwidth}
\includegraphics[scale=0.45,trim={5cm 9.5cm 0cm 10cm}]{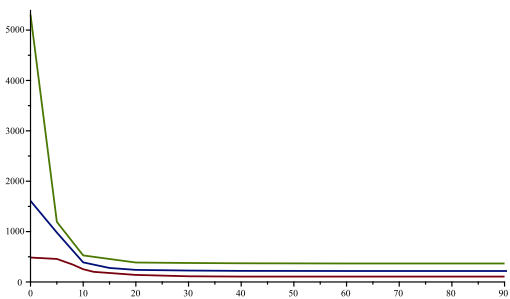}
\vspace*{-4cm}
\caption{Garage number in a map with border}
\end{subfigure}
\hspace*{2cm}
\begin{subfigure}[t]{0.425\textwidth}
\includegraphics[scale=0.45,trim={4cm 9.5cm 0cm 10cm}]{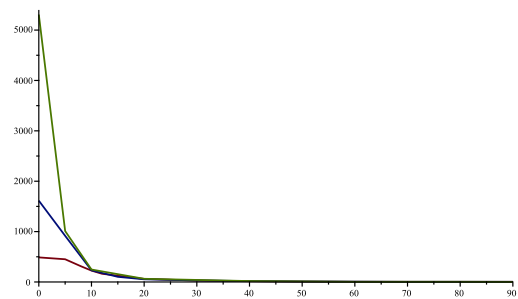}
\vspace*{-4cm}
\caption{in a Torus map with no border}
\end{subfigure}
\vspace{-2cm}
\caption{Size of the garage set as function of coverage radius, in brown $\theta=d_r/4$, in blue $\theta=1.2 d_r/4$, and in green $\theta=1.4 d_r/4$.}
    \label{fig:MPRSize}
\end{figure}

\begin{figure}
\vspace{5cm}
\begin{subfigure}[t]{0.225\textwidth}
\includegraphics[scale=0.55,trim={3cm 9.5cm 0cm 10cm}]{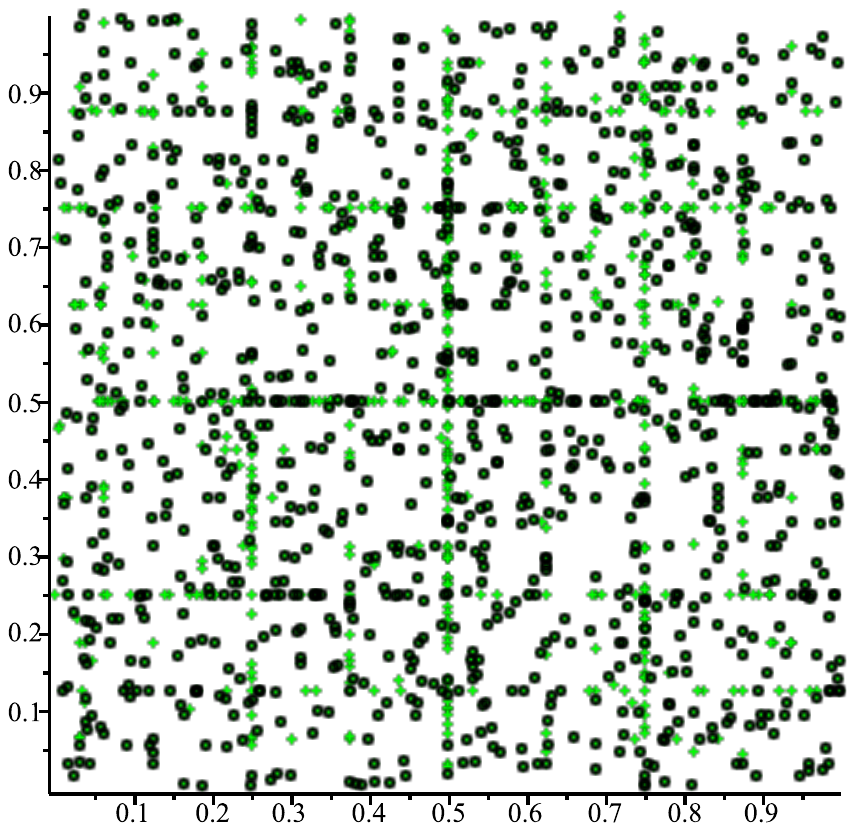}
\vspace*{-4cm}
\caption{coverage distance $5R_n$}
\end{subfigure}
\hspace*{2cm}
\begin{subfigure}[t]{0.225\textwidth}
\includegraphics[scale=0.55,trim={3cm 9.5cm 0cm 10cm}] {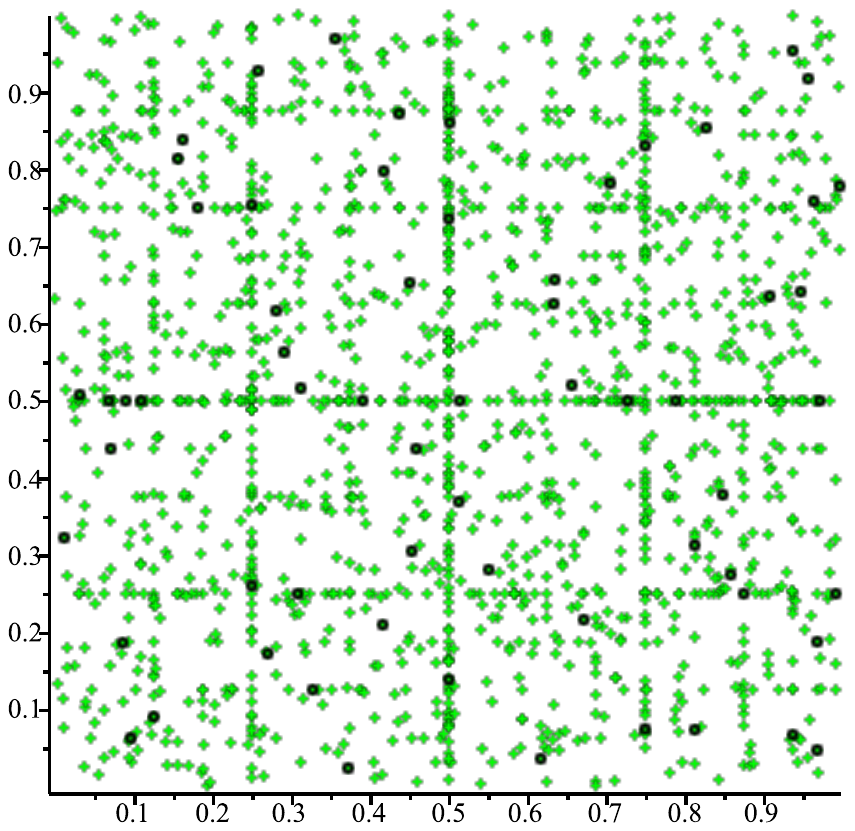}
\vspace*{-4cm}
\caption{coverage distance $20R_n$}
\end{subfigure}
\hspace*{2cm}
\begin{subfigure}[t]{0.225\textwidth}
\includegraphics[scale=0.55,trim={3cm 9.5cm 0cm 10cm}] {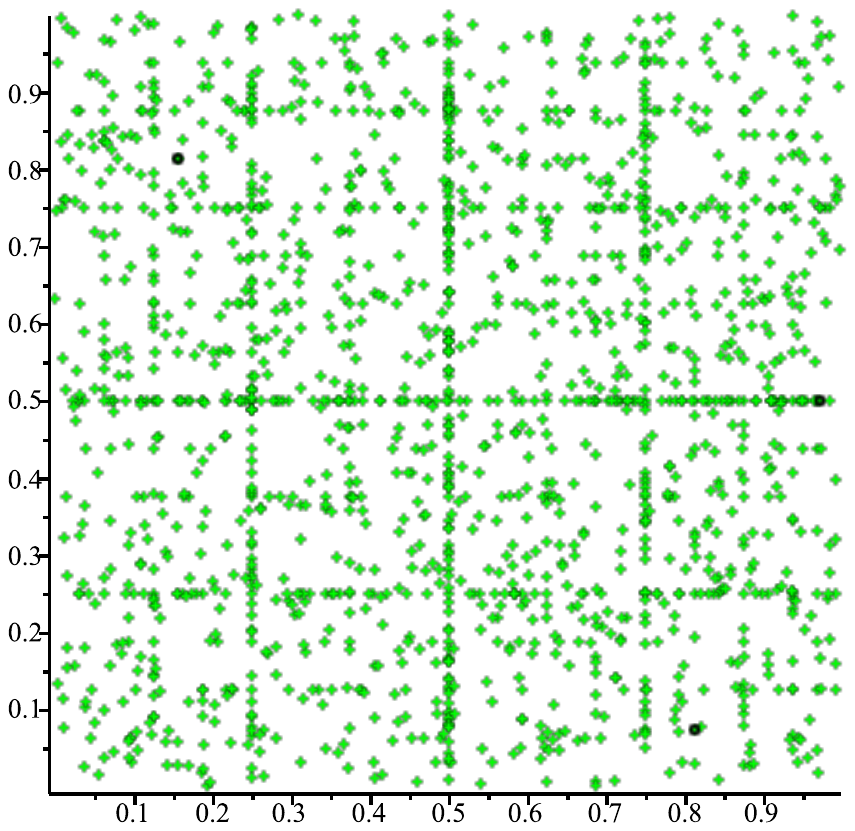}
\vspace*{-4cm}
\caption{coverage distance $80R_n$}
\end{subfigure}
\vspace{-2cm}
\caption{Map of drone garages (black circles), among the base stations (green crosses), for different coverage distances in a torus map.}
\label{fig_MapMPRtorus}
\end{figure}

Figure~\ref{fig_MPRcover} shows the distribution of distance of the mobile nodes to the closest garage. In green is for coverage radius $5R_n$, in blue for coverage radius $10R_n$, in red $20R_n$, in brown $40R_n$, in black $80R_n$. We notice that despite the coverage radius increases the typical distance to the closest garage does not grow too much in comparison because the residual density of garage prevents. Remember that the distance to the closest garage is larger than the number of drones needed to connect the mobile node to the closest relay, which is given by figure~\ref{fig_coverageRelay10}. The distance to the closest garage gives an indication on how fast drones can be moved towards new mobile nodes.

\begin{figure}
\vspace{-1cm}
\begin{subfigure}[t]{0.325\textwidth}
\includegraphics[height=14cm]{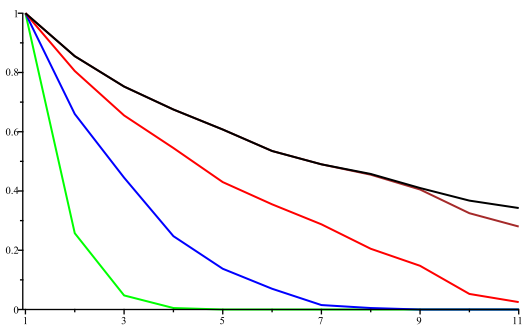}
\vspace{-9cm}
\caption{$n=50,000$, $\theta=1.2d_r/4$}
\end{subfigure}
\hspace{3cm}
\begin{subfigure}[t]{0.325\textwidth}
\includegraphics[height=14cm]{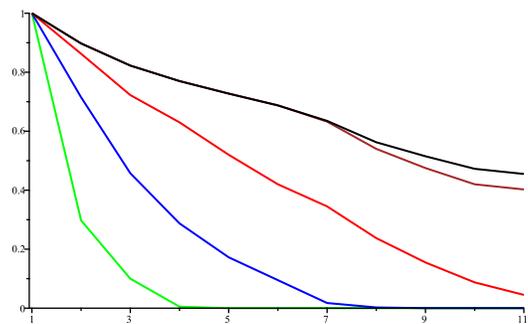}
\vspace{-9cm}
\hspace{3cm}
\caption{$n=50,000$, $\theta=1.4d_r/4$}
\end{subfigure}
\vspace{-6cm}
\caption{distribution of distances to closest garage (in hop count) for various coverage radii.}
\label{fig_MPRcover}
\end{figure}

%
\section{Concluding Remarks} \label{conclusions}

This work has provided a study of the connectivity properties and dimensioning of the moving networks of drones used as a flying backhaul in an urban environment with vehicular users.  

By making use of the hyperfractal model for both the vehicular networks and for the fixed eMBB infrastructures, we have derived analytic bounds on the requirements in terms of connectivity extension: we have proved that for $n$ mobile nodes (distributed according to a hyperfractal distribution of dimension $d_F$) and an average of $\rho$ gNBs (of dimension $d_r$) if $\rho=n^\theta$ with $\theta>d_r/4$, the average fraction of mobile nodes not covered by a gNB tends to zero like $O\left(n^{-\frac{(d_F-2)}{d_r}(2\theta-\frac{d_r}{2})}\right)$. 
Furthermore, for the same regime of  $\theta$, we have obtained that the number  of  drones   to connect the isolated mobile nodes is asymptotically equivalent to the number of isolated mobile nodes. This gives insights on the dimensioning of the flying backhauls and also limitations of the usage of UAVs (second regime of $\theta$).
This work has also initiated the discussions on the placement of the home locations of the drones, what we called the ``garage of drones". We have provided a fast procedure to select the relays that will be garages (and store drones) in order to minimize the number of garages and minimize the delay. 

Our simulations results concur with our bounds, and illustrate the step-change of regime based on $\theta$. Our simulations also show how this can be exploited to have as few garages as possible, while having drones servicing  efficiently the mobile vehicles with limited delay. Hence making the scenario attractive for further study and possible implementation.

Overall our results have provided a realistic stochastic communication model for studying the development of 5G in smart cities. The interest of such an innovative framework was demonstrated by the computation of exact bounds and the identification of particular behaviours (such as the characterisation of a threshold). It is also a step towards constructing a ``smart city modeling'' framework that can be exploited in other urban scenarios.

\section*{Appendix}
Proof that $f(y)=\sum_Hpq^H\exp\left(-p'(q'/2)^Hy\right)=\frac{p(p')^{-\delta}}{\log(2/q')}\Gamma(\delta)y^{-\delta}(1+o(1))$. We use the technique in~\cite{flajolet} by the Mellin transform $f^*(s)=\int_0^\infty f(y)y^{s-1}dy$, which is defined for some complex number $s$ such that $\Re(s)>0$. Indeed since the Mellin transform of $\exp(-p'(q'/2)^H y)$ is $(p'(q'/2)^H)^s\Gamma(s)$ where $\Gamma(s)$ is the Euler ``Gamma" function defined for $\Re(s)>0$, thus $f^*(s)=\sum_H pq^H(p'(q'/2)^H)^{-s}\Gamma(s)=\frac{p(p')^{-s}}{1-q(q'/2)^{-s}}\Gamma(s)$ as long as $\Re(s)<\delta$ (thus the sum $\sum_H pq^H(p'(q'/2)^H)^{-s}$ absolutely converges). 

The asymptotic of function $f(y)$ is obtained by the inverse Mellin transform as explained in~\cite{flajolet} as the residues of function of $f^*(s)y^{-s}$ on the main pole $s=\delta$ which lead to the claimed asymptotic expression. To the risk to be pedantic the reference~\cite{flajolet} also mentions that there are additional poles on the complex numbers $\delta+2ik\pi/\log(q'/2)$ for $k$ integer which lead to negligible fluctuations of the main asymptotic term.

\bibliographystyle{IEEEtran}
\bibliography{mybib}

\end{document}